%% file: _main.tex
\documentclass[a4paper,UKenglish,cleveref, autoref,thm-restate]{lipics-v2021}
\hideLIPIcs  %uncomment to remove references to LIPIcs series (logo, DOI, ...), e.g. when preparing a pre-final version to be uploaded to arXiv or another public repository

\title{New Algorithms for Parity-SAT and Its Bounded-Occurrence Versions} %TODO Please add

\author{Sanjay Jain}{School of Computing, National University of Singapore, Singapore}{sanjayjain@nus.edu.sg}{https://orcid.org/0000-0001-6798-8330}{}

\author{Junqiang Peng}{University of Electronic Science and Technology of China, China}{jqpeng0@foxmail.com}{https://orcid.org/0000-0003-2742-5562}{corresponding author}

\author{Frank Stephan}{Department of Mathematics, National University of Singapore, Singapore\\ School of Computing, National University of Singapore, Singapore}{fstephan@comp.nus.edu.sg}{https://orcid.org/0000-0001-9152-1706}{}

\author{Haoyun Tang}{School of Computing, National University of Singapore, Singapore}{e1154532@u.nus.edu}{https://orcid.org/0009-0001-5410-0221}{}

\author{Mingyu Xiao}{University of Electronic Science and Technology of China, China}{myxiao@uestc.edu.cn}{https://orcid.org/0000-0002-1012-2373}{}

\authorrunning{S.~Jain, J.~Peng, F.~Stephan, H.~Tang, and M.~Xiao} %TODO mandatory. First: Use abbreviated first/middle names. Second (only in severe cases): Use first author plus 'et al.'

\Copyright{Sanjay Jain, Junqiang Peng, Frank Stephan, Haoyun Tang, and Mingyu Xiao} %TODO mandatory, please use full first names. LIPIcs license is "CC-BY";  http://creativecommons.org/licenses/by/3.0/

\ccsdesc[500]{Theory of computation~Design and analysis of algorithms}

\keywords{Parity-SAT, Exact Exponential Algorithms} %TODO mandatory; please add comma-separated list of keywords

\funding{{Junqiang Peng and Mingyu Xiao were supported by the National Natural Science Foundation of China, under the grants 62372095 and 62502078. Sanjay Jain and Frank Stephan were supported by Singapore Ministry of Education (MOE) AcRF Tier 2 grant MOE-000538-00. Additionally, Sanjay Jain was supported by NUS grant E-252-00-0021-01. Additionally, Frank Stephan was supported by AcRF Tier 1 grants A-0008454-00-00 and A-0008494-00-00.}}
\nolinenumbers %uncomment to disable line numbering

\input{_preamble}

\begin{document}

\maketitle

\begin{abstract}

Parity-SAT is the problem of determining whether a given CNF formula has an odd number of satisfying assignments. As a canonical $\oplus$P-complete problem, it represents a fundamental variant of the exact model counting problem (\#SAT). Under the Strong Exponential Time Hypothesis (SETH), Parity-SAT admits no $O^*((2-\varepsilon)^n)$-time or $O^*((2-\varepsilon)^m)$-time algorithm for any constant $\varepsilon>0$, where $n$ and $m$ denote the numbers of variables and clauses, respectively. Thus, breaking the $2^n$ or $2^m$ barrier appears impossible in full generality.

In this work, we revisit this barrier through structural restrictions and a refined exploitation of parity. We study Parity-$d$-occ-SAT, where each variable appears in at most $d$ clauses, and obtain three main results.
First, we design {a randomized} $O^*(2^{m(1-1/O(d))})$-time algorithm, thereby breaking the $2^m$ barrier for every fixed $d$. Second, for the special case $d=2$, we develop a significantly sharper branching algorithm running in $O^*(1.1193^n)$ time or $O^*(1.3248^m)$ time. Third, leveraging the structural insights underlying the $d=2$ case, we obtain an $O^*(1.1052^L)$-time algorithm for general Parity-SAT, where $L$ denotes the formula length.
All algorithms use only polynomial space. Notably, our running-time bounds are better than the best known bounds for the corresponding exact counting counterparts, highlighting a genuine algorithmic advantage of parity over counting. Conceptually, our results demonstrate that parity admits finer structural reductions and more efficient branching than exact model counting, and that bounded occurrence can be systematically leveraged to circumvent classical exponential barriers.

\end{abstract}

\newpage
\input{sections/intro}

\input{sections/prelim}

\input{sections/basic-rules}

\input{sections/Parity-d-occ-SAT-m}

\input{sections/Parity-2-occ-SAT}

\input{sections/Parity-SAT-L}

\input{sections/conclusion}

\bibliography{_ref}

\input{sections/appendix}

\end{document}

%% file: _preamble.tex
\usepackage{amsmath,amsthm,amssymb,mathtools}
\usepackage{xcolor}
\usepackage{booktabs,multirow,multicol,makecell}
\usepackage{todonotes}
\usepackage[noabbrev]{cleveref}
\usepackage{algorithm}
\usepackage[noend]{algpseudocode}
\usepackage{tikz}
\usetikzlibrary{graphs, graphs.standard, calc, positioning, decorations.pathreplacing, shapes.geometric, arrows.meta}
\newcommand{\bigOstar}[1]{O^*(#1)}
\newcommand{\bigO}[1]{{O}(#1)}
\newcommand{\ParitySAT}{Parity-SAT}
\newcommand{\ParitydOccSAT}{Parity-$d$-occ-SAT}
\newcommand{\ParityTwoOccSAT}{Parity-$2$-occ-SAT}
\newtheorem{rrule}{R-Rule}
\DeclareMathOperator{\var}{var}
\DeclareMathOperator{\parity}{Parity}
\DeclarePairedDelimiter{\abs}{\lvert}{\rvert}
\newcommand{\ResultTwoOccN}{1.1193^n}
\newcommand{\ResultTwoOccM}{1.3248^m}

\newcommand{\nl}{\overline}
\newcommand{\formula}{\varphi}
\newcommand{\literal}{\ell}

\newcommand{\LeftComment}[1]{%
  \Statex \hspace*{\dimexpr-\algorithmicindent*\value{ALG@nested}\relax} \textcolor{gray}{$\triangleright$ #1}%
}
\newcommand{\mycomment}[1]{\textcolor{gray}{$\triangleright$ #1}}

%% file: sections/intro.tex
\section{Introduction}

The {\sc Satisfiability} (SAT) problem, determining whether a given Conjunctive Normal Form (CNF) formula has a satisfying assignment, is a cornerstone of Computer Science and the prototypical {\sf NP}-complete problem~\cite{conf/stoc/Cook71}. Its counting variant, {\sc Model Counting} (\#SAT), requires computing the exact number of satisfying assignments (called models). \#SAT is \#{\sf P}-complete~\cite{journals/tcs/Valiant79} and serves as a fundamental tool in probabilistic inference~\cite{journals/ai/Roth96,conf/focs/BacchusDP03,conf/aaai/SangBK05,journals/ai/ChaviraD08}, network reliability estimation~\cite{conf/aaai/Duenas-OsorioMP17}, and explainable AI~\cite{conf/sat/NarodytskaSMIM19}.
A significant modular variant of \#SAT is \ParitySAT{}, which asks whether a formula has an odd number of models. \ParitySAT{} is the canonical complete problem for the complexity class $\oplus${\sf P} introduced by Papadimitriou and Zachos~\cite{conf/tcs/PapadimitriouZ83} and is shown to be NP-hard under randomized reductions by Valiant and Vazirani~\cite{journals/tcs/ValiantV86}. 
This problem also plays a pivotal role in the proof of Toda's theorem~\cite{journals/siamcomp/Toda91}, a landmark result in the study of counting complexity.

Given the inherent hardness of these problems, research has shifted toward finding \emph{moderately exponential-time} algorithms. Typically, instances are parameterized by the number of variables $n$, with the goal of achieving a running time of $\bigOstar{c^n}$ for some constant $c < 2$. A brute-force search over $2^n$ assignments solves SAT, \#SAT, and \ParitySAT{} in $\bigOstar{2^n}$ time.
However, it is not easy to break the trivial ``$2^n$-barrier''.

Under the Strong Exponential Time Hypothesis (SETH)~\cite{journals/jcss/ImpagliazzoP01}, no $\bigOstar{(2-\varepsilon)^n}$ algorithm exists for SAT for any $\varepsilon > 0$. By the isolation lemma~\cite{journals/jcss/CalabroIKP08,conf/iwpec/Traxler08}, this barrier extends to \ParitySAT{} as well. These lower bounds have motivated the study of \emph{restricted} formula classes, such as $k$-SAT (where clauses have size at most $k$) and $d$-occ-SAT (where variables appear in at most $d$ clauses). For $k$-SAT, it has been extensively studied and we provide a review on it later in this section. In this paper,
we focus on $d$-occ-SAT and
denote its counting and parity counterparts by \#$d$-occ-SAT and \ParitydOccSAT{}, respectively.

It is known that $d$-occ-SAT can be solved in $\bigOstar{2^{n(1-{1}/{\bigO{\log d}})}}$ time~\cite{conf/coco/CalabroIP06,series/faia/DantsinH21}.
Remarkably, this upper bound carries over to the exact counting variants: \#$d$-occ-SAT admits a deterministic algorithm running in $\bigOstar{2^{n(1-{1}/{\bigO{\log d}})}}$ time, albeit requiring exponential space~\cite{conf/soda/ChanW16}.
Naturally, this algorithm also applies to \ParitydOccSAT{}.

The computational landscape shifts significantly when instances are parameterized by the number of clauses $m$. While SAT admits an $\bigOstar{1.2226^m}$-time algorithm~\cite{journals/tcs/ChuXZ21}, its parity (and counting) counterparts face a rigid barrier: although they can be solved in $\bigOstar{2^m}$ time via Inclusion-Exclusion~\cite{journals/siamcomp/Iwama89,journals/ipl/Lozinskii92}, Cygan et al.~\cite{journals/talg/CyganDLMNOPSW16} showed that this $2^m$-barrier is insurmountable under SETH even for \ParitySAT{}. Moreover, for the bounded-occurrence restrictions, to the best of our knowledge, no known algorithms break the $2^m$-barrier even for \ParityTwoOccSAT{}.

\subparagraph*{Our Results.}

\begin{table}[t]
    \centering
    \input{tables/result-summary}
    \caption{{Summary of known results and our results. In this table, {`Para.'} stands for `Parameter'; {`det.'} and {`rand.'} stand for deterministic and randomized algorithms, respectively; {`poly-sp.'} and {`exp-sp.'} denote polynomial and exponential space.}}
    \label{tab:summary}
\end{table}

In this work, we study fast algorithms for \ParitySAT{} and \ParitydOccSAT{}. 
{First, we demonstrate that the $2^m$-barrier is surmountable for the class of bounded-occurrence formulas by giving a randomzied $\bigOstar{2^{m(1-1/\bigO{d})}}$-time algorithm for \ParitydOccSAT{}.}
Second, for the specific case of $d=2$, we design a refined algorithm for \ParityTwoOccSAT{} that runs in $\bigOstar{\ResultTwoOccN}$ time or $\bigOstar{\ResultTwoOccM}$ time. Finally, utilizing the results and algorithmic ideas developed for \ParityTwoOccSAT{}, we obtain an $\bigOstar{1.1052^L}$-time algorithm for the general \ParitySAT{}, where $L$ is the formula length (i.e., sum of the sizes of clauses). We remark that the parameter $L$ is also well-studied in the literature, as detailed in the related work discussion later in this section. {All our algorithms require only polynomial space}, and our bounds are better than the known results for their counting counterparts. This underscores that parity counting, being a special case of exact counting, potentially offers greater tractability. A summary of known results and our contributions is provided in Table~\ref{tab:summary}.

Our results mainly rely on the exploitation of parity. A simple example is that if an unassigned variable appears in no clauses, it acts as a {free variable} that doubles the total number of satisfying assignments; consequently, one can immediately infer that the parity of the formula is zero. Building upon this and other parity-specific properties, we develop several reduction rules and branching rules used in our algorithms. Furthermore, unlike most branching algorithms that solely branch on variables, we employ an additional clause-based branching rule, which may be of independent interest.

\subparagraph*{Other Related Works.}
It is known that $k$-SAT can be solved in $\bigOstar{2^{n(1-{1}/{\bigO{k}})}}$ time~\cite{conf/focs/Schoning99,conf/focs/PaturiPZ97,journals/jacm/PaturiPSZ05}. For \#$k$-SAT, the same runtime bound is achieved either deterministically with exponential space~\cite{conf/soda/ChanW16} or via a randomized polynomial-space algorithm~\cite{conf/soda/ImpagliazzoMP12}.
Particular attention has also been given to small values of $k$. Specifically, $3$-SAT admits a randomized $\bigOstar{1.306973^n}$-time~\cite{conf/focs/Scheder21} and a deterministic $\bigOstar{1.32793^n}$-time algorithm~\cite{conf/icalp/Liu18}, and (weighted) \#$2$-SAT is solvable in $\bigOstar{1.2377^n}$ time~\cite{conf/iwpec/Wahlstrom08} or $\bigOstar{1.1082^m}$ time using exponential space~\cite{conf/ijcai/0001S025}.

% ~\cite{journals/iandc/Gelder88,kullmann1997deciding,conf/soda/Hirsch98,journals/jar/Hirsch00a,conf/esa/Wahlstrom05,conf/wads/ChenL09}. 
Beyond $n$ and $m$, the formula length $L$ is another widely studied parameter. Currently, SAT and \#SAT can be solved in $\bigOstar{1.0638^L}$~\cite{journals/iandc/PengX23} and $\bigOstar{1.1082^L}$~\cite{conf/ijcai/0001S025} time, respectively. The latter bound is obtained by reducing a \#SAT instance of length $L$ to a weighted \#2-SAT instance with $m = L$ clauses~\cite{conf/lics/BannachDGH25,conf/sosa/BannachDGH26} and calling the $\bigOstar{1.1082^m}$ time algorithm for weighted \#2-SAT~\cite{conf/ijcai/0001S025}; while the reduction itself only uses polynomial space, the invoked algorithm~\cite{conf/ijcai/0001S025} requires exponential space.

We also highlight that the parity counting versions of many classic problems have been extensively studied, including Perfect Matching~\cite{journals/tcs/Valiant79}, Hamiltonian Cycle~\cite{conf/focs/BjorklundH13}, and $k$-Path~\cite{conf/esa/CurticapeanDH21}. Some of these problems admit polynomial-time algorithms, while others remain computationally hard. Consequently, parity counting problems have attracted significant attention from the perspectives of classical complexity theory, parameterized complexity, and fine-grained complexity (see e.g., \cite{journals/iandc/ArvindK06,conf/focs/Valiant06,conf/esa/CurticapeanDH21,journals/algorithmica/GoldbergR24,conf/icalp/AbboudFW20,conf/icalp/BjorklundDH15}).

%% file: tables/result-summary.tex
\begin{tabular}{ccc|cc}
\toprule
    \textbf{Para.} & \multicolumn{2}{c|}{\textbf{\#SAT}} & \multicolumn{2}{c}{\textbf{\ParitySAT{}}}\\
\midrule
    {$L$} 
    & \multicolumn{2}{c|}{\makecell[c]{$\bigOstar{1.1082^L}$ \\ det., exp-sp. \cite{conf/lics/BannachDGH25,conf/sosa/BannachDGH26,conf/ijcai/0001S025}}}
    & \multicolumn{2}{c}{\makecell[c]{$\bigOstar{1.1052^L}$ \\ det., poly-sp. {[\textbf{\cref{sec:L}}]}}} \\
\bottomrule
\toprule    
    \textbf{Para.} & \textbf{\#2-occ-SAT} & \textbf{\#d-occ-SAT} & \textbf{\ParityTwoOccSAT{}} & \textbf{\ParitydOccSAT{}} \\
\midrule
    \multirow{2}{*}{$n$} 
    & \makecell[c]{$\bigOstar{1.2282^n}$ \\ det., exp-sp.} 
    & \makecell[c]{$\bigOstar{2^{n(1-{1}/{\bigO{\log d}})}}$ \\ rand., poly-sp. \cite{conf/soda/ImpagliazzoMP12}} 
    & \makecell[c]{$\bigOstar{\ResultTwoOccN}$ \\ det., poly-sp.} 
    & \makecell[c]{$\bigOstar{2^{n(1-{1}/{\bigO{\log d}})}}$ \\ rand., poly-sp. \cite{conf/soda/ImpagliazzoMP12}} \\
    
    & ({derived by} $L \le 2n$) 
    & det., exp-sp.~\cite{conf/soda/ChanW16} 
    & {[\textbf{\cref{sec:2-occ}}]} 
    & det., exp-sp.~\cite{conf/soda/ChanW16} \\
\midrule
    \multirow{2}{*}{$m$} 
    & \multicolumn{2}{c|}{\makecell[c]{$\bigOstar{2^m}$ \\ det., poly-sp. \cite{journals/siamcomp/Iwama89,journals/ipl/Lozinskii92}}} 
    & \makecell[c]{$\bigOstar{\ResultTwoOccM}$ \\ det., poly-sp.} 
    & \makecell[c]{$\bigOstar{2^{m(1-{1}/{\bigO{d}})}}$ \\ rand., poly-sp.} \\
    
    & \multicolumn{2}{c|}{} 
    & {[\textbf{\cref{sec:2-occ}}]} 
    & {[\textbf{\cref{sec:d-occ}}]} \\
\bottomrule
\end{tabular}

%% file: sections/prelim.tex
\section{Preliminaries}\label{sec:prelim}
\subparagraph*{Notations.}
Let $V$ be a set of \emph{Boolean variables}. A Boolean variable (or simply \emph{variable}) $x \in V$ can be assigned a value of $1$ (\textsc{True}) or $0$ (\textsc{False}). A variable $x$ has two corresponding \emph{literals}: the positive literal $x$ and the negative literal $\nl{x}$. We use $\nl{x}$ to denote the negation of literal $x$, so $\nl{\nl{x}}=x$.
A \emph{clause} on $V$ is a set of literals on $V$. Note that a clause might be empty. {A \emph{CNF formula} (or simply \emph{formula}) $\formula = \{C_1, C_2, \dots, C_m\}$ over $V$ is a set of clauses on $V$, also denoted by $\formula=\bigwedge_{i=1}^m C_i$.} We denote by $\var(\formula)$ the variable set of $\formula$. For a literal $\literal$, $\var(\literal)$ denotes its corresponding variable. For a clause $C$, $\var(C)$ denotes the set of variables appearing in the clause (i.e., $x \in \var(C)$ if $x \in C$ or $\nl{x} \in C$).
We say a variable $x$ \emph{co-occurs} with a variable $y$ if $x$ and $y$ share a clause.

An \emph{assignment} for a variable set $V$ is a mapping $\sigma: V\rightarrow\{0, 1\}$. Given an assignment $\sigma$, a clause is \emph{satisfied} by $\sigma$ if at least one literal in it evaluates to $1$ under $\sigma$. An assignment for $\var(\formula)$ is called a \emph{satisfying assignment} of $\formula$ if $\sigma$ satisfies all clauses in $\formula$.
We use $\parity(\formula)$ to denote the parity of the number of satisfying assignments of $\formula$ (i.e., it is $1$ if the number of satisfying assignments is odd, and $0$ otherwise).

The \emph{degree} of a variable $x$ in formula $\formula$ is the total number of occurrences of literals $x$ and $\nl{x}$ in $\formula$. A variable is called a \emph{$d$-variable} if its degree is exactly $d$, and a \emph{$d^+$-variable} if its degree is at least $d$. 
% The \emph{degree of a formula} $\formula$, denoted by $\deg(\formula)$, is the maximum degree of all variables in $\formula$.
A variable $x$ is called $(i,j)$-variable if literal $x$ appears $i$ times and literal $\nl{x}$ appears $j$ times.

The \emph{length} of a clause $C$, denoted by $|C|$, is the number of literals in $C$. A clause is called a \emph{$k$-clause} (resp. \emph{$k^-$-clause} and \emph{$k^+$-clause}) if its length is exactly $k$ (resp. at most $k$ and at least $k$). A formula $\formula$ is called a \emph{$k$-CNF formula} if each clause in $\formula$ has length at most $k$.

For a formula $\formula$, define $n(\formula) = \abs{\var(\formula)}$ as the number of variables, $m(\formula) = \abs{\formula}$ as the number of clauses, and $L(\formula) = \sum_{C \in \formula} \abs{C}$ as the total length of the formula. We may simply write $n$, $m$, and $L$ if the context is clear.

\subparagraph*{Branch-and-Search.}
A \emph{branch-and-search} algorithm first applies reduction rules to simplify the instance, and then recursively solves it via branching operations. To evaluate the time complexity, we associate each problem instance with a non-negative measure $\mu$. Let $T(\mu)$ denote an upper bound on the number of leaves in the search tree generated by the algorithm for any instance with measure at most $\mu$. A branching operation generates $l$ subinstances whose solutions collectively determine the solution to the original instance, and the measure decreases by at least $a_i > 0$ in the $i$-th branch, yielding the recurrence relation
\[
    T(\mu) \leq T(\mu - a_1) + \dots + T(\mu - a_l).
\]
This recurrence is represented by a \emph{branching vector} $(a_1, \dots, a_l)$. The \emph{branching factor} of this recurrence, denoted by $\tau(a_1, \dots, a_l)$, is defined as the unique positive real root of the characteristic equation $1 - \sum_{i=1}^l x^{-a_i} = 0$. If the branching factor of every operation in the algorithm is bounded by a constant $\gamma$, then the overall size of the search tree is bounded by $\bigOstar{\gamma^\mu}$. For a comprehensive background on exact exponential algorithms, we refer the reader to~\cite{series/txtcs/FominK10}.
The following standard property allows us to analytically upper-bound branching factors, which is a direct consequence of Lemma~2.2 and Lemma~2.3 in~\cite{series/txtcs/FominK10}.

\begin{lemma}\label{lem:vec-dominate}
    For any reals $a_1, a_2 > 0$ and $p > q > 0$, if $a_1 + a_2 \geq p$ and $\min(a_1, a_2) \geq q$, then $\tau(a_1, a_2) \leq \tau(p-q, q)$.
\end{lemma}

%% file: sections/basic-rules.tex
\section{Reduction Rules and Branching Rules}\label{sec:rules}
In this section, we introduce some basic reduction rules and branching rules that will be frequently used in algorithms in subsequent sections.

Before presenting the rules, we define the notation for the formula simplifications and modifications involved.
For a literal $\literal$, we use $\formula[\literal=1]$ to denote the formula obtained from $\formula$ by assigning $\literal = 1$. This simplifies $\formula$ by removing all satisfied clauses (those containing $\literal$) and deleting all falsified literals (occurrences of $\nl{\literal}$) from the remaining clauses. Accordingly, we define $\formula[\literal=0]=\formula[\nl{\literal}=1]$.
Extending this to a clause $C=(\literal_1\lor \cdots \lor \literal_k)$, we define $\formula[C=0] = \formula[\literal_1=0, \dots, \literal_k=0]$ (falsifying $C$ by assigning all its literals to $0$) and $\formula[C=1] = \formula \land C$ (requiring $C$ to be satisfied by adding it to the formula).
We use a combined notation to represent a sequence of operations {in order} (e.g., $\formula[C_1=1, C_2=0, \literal=1]$ denotes the formula obtained by adding $C_1$, falsifying $C_2$, and assigning $\literal=1$).

\subsection{Reduction Rules}\label{subsec:reduction-rules}

A reduction rule (R-Rule) is correct if it takes a formula $\formula$ as input and outputs a formula $\formula'$ with $\parity(\formula) = \parity(\formula')$, or correctly reports the value of $\parity(\formula)$.
{The correctness of most rules is easily seen, and we will briefly justify some of them. A formal statement of correctness is given in \cref{lem:rules-correctness} below.}
{The first five reduction rules are standard, which can also be applied to exact counting.}
\begin{rrule}[Base case]\label{rrule:empty-cls}
    If formula $\formula$ contains an empty clause, report $\parity(\formula)=0$.
\end{rrule}

\begin{rrule}[Elimination of duplicated literals]\label{rrule:duplicated}
    If a clause $C$ contains duplicated literals $\literal$, remove all but one $\literal$ in $C$.
\end{rrule}

\begin{rrule}[Elimination of tautology]\label{rrule:taut}
    If a clause $C$ contains two complementary literals $\literal$ and $\nl{\literal}$, remove clause $C$.
\end{rrule}

\begin{rrule}[Elimination of subsumptions]\label{rrule:subsumption}
    If there are two clauses $C$ and $D$ such that $C\subseteq D$, remove clause $D$.
\end{rrule}

\begin{rrule}[Elimination of $1$-clauses]\label{rrule:1cls}
    If there is a $1$-clause $(\literal)$, assign $\literal=1$.
\end{rrule}

{Next, we introduce four rules (i.e., R-Rules~\ref{rrule:0var}, \ref{rrule:1var}, \ref{rrule:dominate}, \ref{rrule:twins}) that are specific for \ParitySAT{}.}
\begin{rrule}[Elimination of $0$-variables]\label{rrule:0var}
    If there is an unassigned variable $x$ that does not appear in any clause, report that $\parity(\formula)=0$.
\end{rrule}

The correctness of \cref{rrule:0var} follows from a simple parity-based observation: if there is a $0$-variable, then the parity of the number of satisfying assignments is even. This is because such a {free variable} doubles the number of solutions. 

We say that a literal $x$ \emph{dominates} a variable $y$ (where $y \neq \var(x)$) if every clause containing variable $y$ also contains literal $x$. Observe that assigning $x=1$ makes $y$ a $0$-variable in the resulting formula $\formula[x=1]$. This implies that $\parity(\formula[x=1])=0$ and so $\parity(\formula)=\parity(\formula[x=0])$. Thus, we have the following two reduction rules:

\begin{rrule}[Elimination of $1$-variables]\label{rrule:1var}
    If there is a $1$-variable $x$ and let $(x\vee C)$ be the clause containing $x$, then assign $x=1$ and assign all literals in $C$ to $0$.
\end{rrule}

\begin{rrule}[Domination]\label{rrule:dominate}
    If there is a literal $x$ that dominates some variable $y$, assign $x=0$.
\end{rrule}

{
Technically, \cref{rrule:1var} is a special case of \cref{rrule:dominate}. If variable $x$ solely appears in $(x \vee C)$, \cref{rrule:dominate} assigns all literals in $C$ to $0$ (as they dominate $x$), leaving $(x)$ and forcing $x=1$ via \cref{rrule:1cls}. We state \cref{rrule:1var} explicitly for conceptual clarity and frequent subsequent use.
}

{
Two literals $\literal_a$ and $\literal_b$ (where $\var(\literal_a) \neq \var(\literal_b)$) are \emph{twins} if every clause in $\formula$ contains $\literal_a$ (resp., $\nl{\literal_a}$) if and only if it also contains $\literal_b$ (resp., $\nl{\literal_b}$).
}

\begin{rrule}\label{rrule:twins}
    If two literals $\literal_a$ and $\literal_b$ are twins in $\formula$, remove the variable $\var(\literal_b)$ from the {variable set (i.e., remove $\literal_b$ and $\nl{\literal_b}$ from the clauses, and $\var(\formula)\gets \var(\formula)\setminus\{\var(\literal_b)\}$).}
\end{rrule}

{The following two rules, \cref{rrule:subsumption-complementary-ltr} and \cref{rrule:two-2cls-complementary}, are also applicable for exact counting.}
\begin{rrule}\label{rrule:subsumption-complementary-ltr}
    If there is a clause $(\literal \vee C)$ and another clause $(\nl{\literal} \vee C \vee D)$, remove the literal $\nl{\literal}$ from $(\nl{\literal} \vee C \vee D)$.
\end{rrule}

For two literals $\literal_a$ and $\literal_b$, \emph{setting $\literal_a=\literal_b$} means replacing all occurrences of $\literal_a$ with $\literal_b$ (and $\nl{\literal_a}$ with $\nl{\literal_b}$) in the formula, and then removing the variable $\var(\literal_a)$ from the variable set.

\begin{rrule}\label{rrule:two-2cls-complementary}
    If there are two $2$-clauses $(x\vee y)$ and $(\nl{x}\vee \nl{y})$ in $\formula$, set $x=\nl{y}$ and apply \cref{rrule:taut} exhaustively. 
\end{rrule}

We use $\textsf{brute-force}(\formula)$ to denote the brute-force algorithm that returns the parity of the number of satisfying assignments of formula $\formula$ by enumerating all possible assignments. Note that $\textsf{brute-force}(\formula)$ runs in $\bigO{1}$ time if $\formula$ has only a constant number of variables.

\begin{rrule}[isolate]\label{rrule:small-subf-0}
    If formula $\formula$ can be partitioned into two non-empty subformulas $\formula_1$ and $\formula_2$, {i.e., $\formula=\formula_1\wedge \formula_2$}, with $\var(\formula_1) \cap \var(\formula_2) = \emptyset$ and $\abs{\var(\formula_1)} \le 10$, do the following:
    (1) $p \leftarrow \textsf{brute-force}(\formula_1)$; (2) {If $p = 0$, report $\parity(\formula)=0$};
    (3) Remove $\formula_1$ from $\formula$. 
\end{rrule}

\begin{rrule}[semi-isolate]\label{rrule:small-subf-1}
    If formula $\formula$ can be partitioned into two non-empty subformulas $\formula_1$ and $\formula_2$, {i.e., $\formula=\formula_1\wedge \formula_2$}, with $\var(\formula_1) \cap \var(\formula_2) = \{x\}$ and $\abs{\var(\formula_1)} \le 10$, do the following:
    (1) $p_1 \leftarrow \textsf{brute-force}(\formula_1[x=1])$ and $p_0 \leftarrow \textsf{brute-force}(\formula_1[x=0])$;
    (2) {If $p_1 = p_0 = 0$, report $\parity(\formula)=0$. } 
    (3) Remove $\formula_1$ from $\formula$; (4) If $p_0 = 1$ and $p_1 = 0$, assign $x=0$. If $p_0 = 0$ and $p_1 = 1$, assign $x=1$. If $p_0 = p_1 = 1$, leave $x$ as an unassigned variable.
\end{rrule}
{Since $\formula_1$ and $\formula_2$ share only the variable $x$, they become disjoint once $x$ is assigned. Thus, the parity calculation splits over the two possible assignments to $x$: $\parity(\formula) \equiv \parity(\formula_1[x=0]) \cdot \parity(\formula_2[x=0]) + \parity(\formula_1[x=1]) \cdot \parity(\formula_2[x=1]) \pmod 2$. Substituting $p_0$ and $p_1$, we obtain $\parity(\formula) \equiv p_0 \cdot \parity(\formula_2[x=0]) + p_1 \cdot \parity(\formula_2[x=1]) \pmod 2$. The four cases directly evaluate this equation for all possible combinations of $(p_0, p_1)$.}

{
We remark that the threshold number 10 in the above two rules is
chosen arbitrarily to exclude small isolated or semi-isolated subformulas in our subsequent analysis. It only needs to be a sufficiently large constant; using a larger or smaller safe value will not affect the current analysis.}

\begin{definition}[reduced formulas]
    We use $R(\formula)$ to denote the resulting formula after exhaustively applying the above reduction rules on $\formula$. A formula $\formula$ is \emph{reduced} if $\formula=R(\formula)$.
\end{definition}

\begin{lemma}\label{lem:rules-correctness}
    All the reduction rules are correct, i.e., $\parity(\formula)=\parity(R(\formula))$.
\end{lemma}
\begin{proof}
    The correctness of the first five basic rules (i.e., elimination of empty clauses, tautologies, duplicated literals, subsumptions, and $1$-clauses) is straightforward.
    The correctness of \cref{rrule:0var}, \cref{rrule:dominate}, \cref{rrule:1var} (i.e., elimination of $0$-variables, domination, $1$-variables) has been justified inline. We prove the correctness of the remaining rules:

    \textbf{\cref{rrule:twins}} removes $\var(\literal_b)$ when two literals $\literal_a$ and $\literal_b$ are twins (i.e., they always appear together with the same polarity). Let $\mathcal{S}$ be the set of clauses in $\formula$ that contain $\var(\literal_a)$ and $\var(\literal_b)$. By the definition of twins, for every clause $C$ in $\mathcal{S}$, either $\literal_a, \literal_b\in C$ or $\nl{\literal_a},\nl{\literal_b}\in C$. We partition the satisfying assignments of $\formula$ based on the four possible evaluations of $(\literal_a, \literal_b) \in \{0, 1\}^2$. Consider the cases $(\literal_a=0, \literal_b=1)$ and $(\literal_a=1, \literal_b=0)$. Under both assignments, all clauses in $\mathcal{S}$ are satisfied. Since $\var(\literal_a)$ and $\var(\literal_b)$ do not appear in any clauses outside $\mathcal{S}$, the resulting formulas under these two assignments are identical. Thus, these two cases cancel each other out in terms of parity. The parity of the formula is then determined by the remaining cases where $\literal_a = \literal_b$ (i.e., $\literal_a=\literal_b=0$ and $\literal_a=\literal_b=1$). In this case, $\literal_a \vee \literal_b \equiv \literal_a$ (and $\nl{\literal_a} \vee \nl{\literal_b} \equiv \nl{\literal_a}$). Therefore, removing $\var(\literal_b)$ safely preserves the parity of the formula.
        
    \textbf{\cref{rrule:subsumption-complementary-ltr}} simplifies a clause $(\nl{\literal} \vee C \vee D)$ to $(C \vee D)$ in the presence of another clause $(\literal \vee C)$. Let $\formula'$ be the formula obtained by this replacement. It suffices to show that an assignment satisfies $\formula$ if and only if it also satisfies $\formula'$. First, any assignment satisfying $\formula'$ satisfies $(C \vee D)$, which also satisfies $(\nl{\literal} \vee C \vee D)$. Thus, it also satisfies $\formula$. Conversely, assume for contradiction that an assignment satisfies $\formula$ but not $\formula'$. This assignment must satisfy $(\nl{\literal} \vee C \vee D)$ and falsify $(C \vee D)$. This implies that $\literal = 0$ (since $C$ and $D$ evaluate to $0$). However, such an assignment would falsify the existing clause $(\literal \vee C)$ in $\formula$. A contradiction.
        
    \textbf{\cref{rrule:two-2cls-complementary}} sets $x = \nl{y}$ when the formula contains two $2$-clauses $(x \vee y)$ and $(\nl{x} \vee \nl{y})$. Observe that setting $x=y=0$ falsifies $(x \vee y)$, and setting $x=y=1$ falsifies $(\nl{x} \vee \nl{y})$. Therefore, any satisfying assignment must satisfy $x \neq y$, which justifies the replacement.
        
    \textbf{\cref{rrule:small-subf-0} and \cref{rrule:small-subf-1}} evaluate and eliminate a small subformula $\formula_1$ that shares at most one variable with the remaining formula $\formula_2$. For \cref{rrule:small-subf-0}, since $\var(\formula_1) \cap \var(\formula_2) = \emptyset$, the satisfying assignments of $\formula = \formula_1 \land \formula_2$ are formed by the Cartesian product of the satisfying assignments of $\formula_1$ and $\formula_2$. Therefore, we have $\parity(\formula) = \parity(\formula_1) \cdot \parity(\formula_2)$. Since $p = \parity(\formula_1)$, if $p = 0$, it directly implies $\parity(\formula) = 0$. If $p = 1$, then $\parity(\formula) = \parity(\formula_2)$. The rule correctly removes $\formula_1$ (leaving $\formula_2$), preserving the parity of the formula.
        
    Consider \cref{rrule:small-subf-1}. By branching on the shared variable $x \in \{0, 1\}$, we partition the satisfying assignments of $\formula$. Since $\formula_1$ and $\formula_2$ share only $x$, assigning a value $b$ to $x$ yields two variable-disjoint subformulas $\formula_1[x=b]$ and $\formula_2[x=b]$. Thus, we have
    \begin{align*}
        \parity(\formula) &= \big( \parity(\formula_1[x=0]) \cdot \parity(\formula_2[x=0]) \big) \oplus \big( \parity(\formula_1[x=1]) \cdot \parity(\formula_2[x=1]) \big)\\
        &= \big(p_0 \cdot \parity(\formula_2[x=0])\big) \oplus \big(p_1 \cdot \parity(\formula_2[x=1])\big).
    \end{align*}
    We analyze the four possible outcomes for $(p_0, p_1) \in \{0, 1\}^2$. If $p_0 = p_1 = 0$, then $\parity(\formula) = 0 \oplus 0 = 0$, and the rule correctly reports $0$. If $p_0 = 1$ and $p_1 = 0$, then $\parity(\formula) = \parity(\formula_2[x=0])$; the rule assigns $x=0$ and leaves $\formula_2[x=0]$, preserving the parity. Similarly, if $p_0 = 0$ and $p_1 = 1$, then $\parity(\formula) = \parity(\formula_2[x=1])$, which is preserved by assigning $x=1$. Finally, if $p_0 = p_1 = 1$, then $\parity(\formula) = \parity(\formula_2[x=0]) \oplus \parity(\formula_2[x=1])$. This is exactly the parity of $\formula_2$ with $x$ treated as an unassigned variable; thus, the rule simply removes $\formula_1$, which correctly leaves $\formula_2$ intact.
\end{proof}

\begin{lemma}\label{lem:poly-time-reduce}
    It takes polynomial time and space to transform any formula $\formula$ to $R(\formula)$.
\end{lemma}
\begin{proof}
    We define a potential function $\Phi(\formula) = (n(\formula), m(\formula), L(\formula))$ and evaluate it using lexicographical order. Observe that every application of a reduction rule decreases $\Phi(\formula)$: 
    \begin{itemize}
        \item Rules~\ref{rrule:1cls}, \ref{rrule:0var}, \ref{rrule:1var}, \ref{rrule:dominate}, \ref{rrule:twins}, \ref{rrule:two-2cls-complementary}, \ref{rrule:small-subf-0}, and \ref{rrule:small-subf-1} decrease $n(\formula)$ by at least $1$.
        \item Rules~\ref{rrule:empty-cls}, \ref{rrule:taut}, \ref{rrule:subsumption}, and  do not increase $n(\formula)$ but decrease $m(\formula)$ by at least $1$.
        \item Rules~\ref{rrule:duplicated} and \ref{rrule:subsumption-complementary-ltr} do not increase $n(\formula)$ or $m(\formula)$, but decrease $L(\formula)$ by at least $1$.
    \end{itemize}
    Since all parameters are bounded by the initial input size and checking the applicability of any rule takes polynomial time and space, the exhaustive application of these rules must terminate in polynomial time and space.
\end{proof}

\begin{lemma}[properties of reduced formula]\label{lem:reduced-formula}
    For a reduced formula $\formula$, it holds that 
    \begin{enumerate}
        \item All variables in $\formula$ are $2^+$-variables;
        \item All clauses in $\formula$ are $2^+$-clauses;
        \item Any two clauses in $\formula$ have at most one common $2$-variable;
        \item Any two $2$-clauses can have at most one common variable;
        \item\label{prop:reduced-formula-subformula} For any subformula $\formula'$ with $\abs{\var(\formula')}\leq 10$, it holds that $\abs{\var(\formula')\cap \var(\formula\setminus \formula')}\geq 2$.
    \end{enumerate}
\end{lemma}
\begin{proof}
    {
    We prove each property by showing that its violation would trigger a reduction rule:
    \begin{enumerate}
        \item If there is a $0$-variable or a $1$-variable, \cref{rrule:0var} or \cref{rrule:1var} would apply, respectively.
        \item Any empty clause or $1$-clause would immediately trigger \cref{rrule:empty-cls} or \cref{rrule:1cls}.
        \item Assume two clauses $C$ and $D$ share two $2$-variables, $x$ and $y$. Note that $x$ and $y$ only appear in $C$ and $D$. We analyze their polarity configurations: (i) If they share the same polarity in both clauses (e.g., $x, y \in C$ and $x, y \in D$, or $x, y \in C$ and $\nl{x}, \nl{y} \in D$), they are twins by definition, triggering \cref{rrule:twins}. (ii) If one variable maintains polarity while the other flips (e.g., $x, y \in C$ and $\nl{x}, y \in D$), then the literal $y$ appears in every clause containing $x$, meaning $y$ dominates $x$, triggering \cref{rrule:dominate}.
        \item Assume two $2$-clauses $C$ and $D$ share two variables $x$ and $y$ and $C=(x\lor y)$. Their structures are restricted to combinations of $x$ and $y$. If $D=(x\lor y)$, \cref{rrule:subsumption} applies. If $D=(\nl{x} \lor y)$ or $D=(x \lor \nl{y})$), \cref{rrule:subsumption-complementary-ltr} is applicable. If $D=(\nl{x} \lor \nl{y})$), \cref{rrule:two-2cls-complementary} applies. 
        \item If the intersection size is $0$, \cref{rrule:small-subf-0} applies. If the intersection size is $1$, \cref{rrule:small-subf-1} applies. Therefore, the intersection must contain at least $2$ variables. \qedhere
    \end{enumerate}
    }
\end{proof}

\subsection{Branching Rules}
Let $x$ be a variable in formula $\formula$.
The simplest branching rule is to branch on $x=0$ and $x=1$, following from $\parity(\formula)=\parity(\formula[x=0])\oplus \parity(\formula[x=1])$; we refer to this as \emph{simple branching}.
We also utilize the following branching rule based on a variable:
\begin{lemma}[variable branching]\label{lem:variable-branching}
    Let $x$ be a $d$-variable in formula $\formula$ that is contained in clauses $(\literal_{x,1}\vee C_1), (\literal_{x,2}\vee C_2), ..., (\literal_{x,d}\vee C_d)$, where $\literal_{x,i}\in \{x,\nl{x}\}$ for each $i\in [d]$.
    Then
    \[
        \parity(\formula)=\bigoplus_{i=1}^{d}\parity(\formula[C_1=\cdots =C_{i-1}=1, C_i=0, \literal_{x,i}=1]).
    \]
\end{lemma}
\begin{proof}
    Consider the partition of the satisfying assignments based on the smallest index $i \in [d]$ such that the sub-clause $C_i$ evaluates to $0$. 
    For a fixed $i$, the condition $(C_1=1, \dots, C_{i-1}=1, C_i=0)$ forces the literal $\literal_{x,i}$ to be $1$ to satisfy the $i$-th clause.
    The only assignments not covered by these $d$ cases are those where $C_1=\dots=C_d=1$. 
    In this remaining case, all clauses containing $x$ are satisfied, making $x$ a $0$-variable. 
    By \cref{rrule:0var}, we have $\parity(\formula[C_1=\cdots=C_d=1])=0$.
    This completes the proof.
\end{proof}

Apart from branching on variables, we can also perform branching based on clauses.
\begin{lemma}[clause branching]\label{lem:clause-branching}
    Let $C$ be a clause in formula $\formula$. Let $\formula \setminus \{C\}$ denote the formula obtained by removing $C$ from $\formula$. Then
    \[
        \parity(\formula) = \parity(\formula\setminus \{C\}) \oplus \parity((\formula\setminus\{C\})[C=0]).
    \]
\end{lemma}
\begin{proof}
    Let $S$ be the set of satisfying assignments for the relaxed formula $\formula \setminus \{C\}$.
    We partition $S$ into two disjoint sets: $S_{\rm sat} = \{\sigma \in S \mid \sigma(C)=1\}$ and $S_{\rm unsat} = \{\sigma \in S \mid \sigma(C)=0\}$.
    Observe that $S_{\rm sat}$ is exactly the set of satisfying assignments for $\formula$, and $S_{\rm unsat}$ corresponds to the satisfying assignments for $(\formula\setminus\{C\})[C=0]$.
    Since $|S| = |S_{\rm sat}| + |S_{\rm unsat}|$, the lemma follows by taking modulo 2.
\end{proof}

{In what follows, we assume any branched formula $\formula$ is reduced; a formula $\formula'$ created in a sub-branch is not necessarily reduced, so we use $R(\formula')$ to denote its reduced one.}

%% file: sections/Parity-d-occ-SAT-m.tex
\section{Faster-than-$2^m$ Algorithms for \ParitydOccSAT{}}\label{sec:d-occ}
In this section, we present algorithms for \ParitydOccSAT{} that improve upon the $2^m$-barrier, where $m$ is the number of clauses. 
To this end, we first give a Turing reduction {from the general \ParitySAT{} (which includes \ParitydOccSAT{}) to its restricted version} on \emph{positive formulas}, i.e., formulas containing only positive literals. 
{We note that \ParitySAT{} on positive formulas cannot be solved in subexponential time under the randomized Exponential Time Hypothesis (see Appendix~\ref{sec:rETH-hard} for details).}

\begin{lemma}\label{lem:reduction-positive}
    Suppose there is an algorithm $\mathcal{A}_{\text{pos}}$ that solves {\ParitySAT{}} on positive formulas in time $\bigOstar{r^m}$. 
    Then, {\ParitySAT{}} can be solved in time
    $
        \bigOstar{\max(1.6181, r)^m}
    $.
\end{lemma}
\begin{proof}
    First, for any variable $x$ that appears solely as negative literals $\nl{x}$ in $\formula$, we flip the variable (i.e., replace every occurrence of $\nl{x}$ with $x$).
    If the formula $\formula$ {(with variables flipped)} is already a positive formula, we directly apply algorithm $\mathcal{A}_{\text{pos}}$ to solve it in $\bigOstar{r^m}$ time.

    Otherwise, there must exist a variable $x$ that has both positive and negative literals, i.e., $x$ is a $(a,b)$-variable with $a,b\ge 1$. Let $C$ be a clause containing the literal $x$. There exists at least one other clause $D$ containing the complementary literal $\nl{x}$. We apply the clause branching rule in \cref{lem:clause-branching} on $C$.
    In the first branch ($\formula\setminus \{C\}$), the clause $C$ is removed, so the number of clauses decreases by $1$.
    In the second branch ($\formula[C=0]$), all literals in $C$ are assigned to $0$. In particular, $x$ is assigned to $0$ and clause $D$ is satisfied. Thus, both $C$ and $D$ are removed. The number of clauses decreases by at least $2$.
    This yields a branching factor of $\tau(1,2)<1.6181$.
    The overall running time is therefore bounded by $\bigOstar{\max(1.6181, r)^m}$.
\end{proof}

The above lemma suggests that, to break the $2^m$-barrier, it suffices to solve this restricted version efficiently.
In our case, we will use the following corollary:
\begin{corollary}\label{coro:reduce-to-positive}
    Suppose there exists an algorithm $\mathcal{A}_{\text{pos}}$ that solves {\ParitydOccSAT{}} on positive formulas in time $\bigOstar{r_d^m}$. Then, {\ParitydOccSAT{}} can be solved in time
    $
        \bigOstar{\max(1.6181, r_d)^m}
    $.
\end{corollary}

Now we are left with solving \ParitydOccSAT{} on positive formulas.
We observe that applying the variable branching from \cref{lem:variable-branching} to positive formulas yields a running time of $\bigOstar{2^{m(1 - 1/\bigO{2^d})}}$.
Specifically, applying this rule on a variable of degree $d$ results in a measure decrease of $d - i + 1$ in the $i$-th branch (where $d$ clauses are satisfied and $i - 1$ clauses are added back). 
This generates a branching factor of $\phi_d=\tau(d, d-1, \dots, 1)$.
Here, $\phi_d$ is the unique real solution of $x^d(2-x)=1$ in the interval $(1,2)$, known as the $d$-th order Fibonacci constant, which is $2^{(1-1/\bigO{2^d})}$.
With this simple algorithm for positive formulas, we conclude that {\ParitydOccSAT{}} can be solved in time $\bigOstar{\phi_d^m} \subseteq \bigOstar{2^{m(1-1/\bigO{2^d})}}$.
That is, for any fixed $d$, the running time has a base strictly smaller than 2.

Next, we demonstrate that by leveraging existing results, we can improve the asymptotic bound to $\bigOstar{2^{m(1 - 1/O(d))}}$, as desired.
\begin{theorem}
    There is a randomized algorithm that solves \ParitydOccSAT{} with $m$ clauses in $\bigOstar{2^{m(1-1/O(d))}}$ time and polynomial space.
\end{theorem}
\begin{proof}
    By \cref{coro:reduce-to-positive}, it suffices to give an $\bigOstar{2^{m(1-1/\bigO{d})}}$-time algorithm for {\ParitydOccSAT{}} on positive formulas. To this end, we reduce this problem with $m$ clauses to \#$d$-SAT with $N=m$ variables {via a chain of reductions}, and then apply the randomized algorithm in \cite{conf/soda/ImpagliazzoMP12} to solve the instance in $\bigOstar{2^{N(1-1/\bigO{d})}}$ time and polynomial space.

    We first formalize the intermediate problems that will be used in the chain of reduction. Given a universe $U$ and a family $\mathcal{F}$ of subsets of $U$, a \emph{hitting set} of $(U, \mathcal{F})$ is a subset $H \subseteq U$ such that for every $S \in \mathcal{F}$, $H \cap S \neq \emptyset$, while a \emph{set cover} of $(U, \mathcal{F})$ is a subcollection $\mathcal{C} \subseteq \mathcal{F}$ such that $\bigcup_{S \in \mathcal{C}} S = U$. 
    The {Parity-HS} (resp.\ {Parity-SC}) problem asks to compute the parity of the number of hitting sets (resp.\ set covers) of $(U, \mathcal{F})$. 
    We denote by {Parity-$d$-occ-HS} the restriction of {Parity-HS} where each element $u \in U$ occurs in at most $d$ sets of $\mathcal{F}$. 
    Furthermore, we define {Parity-$d$-HS} and {Parity-$d$-SC} as the restrictions of {Parity-HS} and {Parity-SC}, respectively, where each set $S \in \mathcal{F}$ has size at most $d$.
    
    {
    An overview of the chain of reductions is illustrated in \cref{fig:reduction-chain}. 
    \begin{figure}[htbp]
        \centering
        \input{figures/diagram_reduction.tex}
        \caption{The chain of reductions from \ParitydOccSAT{} on positive formulas to \#$d$-SAT.}
        \label{fig:reduction-chain}
    \end{figure}
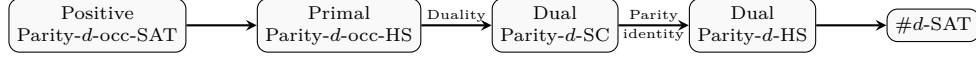
    }

    Consider a positive formula $\formula$ with variable set $V$ and clause set $\mathcal{C}$ (where $|\mathcal{C}|=m$). We can view $\formula$ as a set system $(V, \mathcal{C})$. 
    Since each variable appears at most $d$ times, computing $\parity(\formula)$ is equivalent to solving {Parity-$d$-occ-HS} on $(V, \mathcal{C})$. 
    To solve this, we consider the \emph{dual} set system $(\mathcal{C}, V^*)$, where the universe is the set of clauses $\mathcal{C}$, and $V^*$ contains, for each variable $x \in V$, the set of clauses containing $x$. 
    Crucially, hitting sets in the primal system $(V, \mathcal{C})$ correspond bijectively to set covers in the dual system $(\mathcal{C}, V^*)$.
    Moreover, the frequency constraint in the primal system implies that every set in the dual family $V^*$ has size at most $d$. 
    Thus, the problem is equivalent to {Parity-$d$-SC} on a universe of size $m$.

    It was shown by Cygan et al.~\cite[Theorem 4.3]{journals/talg/CyganDLMNOPSW16} that for any set system, the parity of the number of set covers equals the parity of the number of hitting sets. 
    Applying this identity to our dual system, we are left with solving {Parity-$d$-HS} on the universe $\mathcal{C}$ of size $m$. 
    Finally, this naturally reduces to solving \#$d$-SAT with $m$ variables, which completes the proof.
\end{proof}

%% file: figures/diagram_reduction.tex
\begin{tikzpicture}[
            node distance=0.93cm, % 适当的节点间距
            box/.style={draw, rectangle, rounded corners, align=center, font=\scriptsize, inner sep=3pt, fill=gray!5},
            arrow/.style={-stealth, thick}
            ]
            
            % 节点定义
            \node[box] (step1) {Positive \\ \ParitydOccSAT{}};
            \node[box, right=of step1] (step2) {Primal \\ {Parity-$d$-occ-HS}};
            \node[box, right=of step2] (step3) {Dual \\ {Parity-$d$-SC}};
            \node[box, right=of step3] (step4) {Dual \\ {Parity-$d$-HS}};
            \node[box, right=of step4] (step5) {\#$d$-SAT};
            
            % 带有文字标注的箭头
            \draw[arrow] (step1) -- (step2);
            \draw[arrow] (step2) -- node[above, font=\tiny, inner sep=1pt] {Duality} (step3);
            
            % 分开放置在箭头上下
            \draw[arrow] (step3) -- 
                node[above, font=\tiny, inner sep=1pt] {Parity} 
                node[below, font=\tiny, inner sep=1pt] {identity} (step4);
                
            \draw[arrow] (step4) -- (step5);
            
        \end{tikzpicture}

%% file: sections/Parity-2-occ-SAT.tex
\section{An Algorithm for \ParityTwoOccSAT{}}\label{sec:2-occ}
In this section, we present a fast algorithm for \ParityTwoOccSAT{}. 
We analyze the time complexity of this algorithm with respect to two parameters: the number of variables $n$ and the number of clauses $m$, showing that \ParityTwoOccSAT{} can be solved in $\bigO{\ResultTwoOccN}$ time or $\bigOstar{\ResultTwoOccM}$ time.
{
We remark that \ParityTwoOccSAT{} typically lies in the domain where $m\le n$. Let $\bar{k}$ be the average clause length; then the formula length is $\bar{k}m\le 2n$. Since unit clauses can be eliminated (e.g., by our \cref{rrule:1cls}), we have $\bar{k}\ge 2$, which implies $m\le n$. Thus, a runtime bound of $\bigOstar{2^{cn}}$ for some constant $c$ is not as good as a bound of $\bigOstar{2^{c m}}$ for \ParityTwoOccSAT{}.
}
Besides, we note that this problem cannot be solved in subexponential time under the randomized Exponential Time Hypothesis (see Appendix~\ref{sec:rETH-hard} for details).

Our algorithm consists of two steps that utilize the clause branching in \cref{lem:clause-branching}.
In the first step, we deal with $4^+$-clauses, and directly branch on a clause with the largest size. 
In the second step, all clauses have a size of at most $3$, and we employ a branching strategy that aims to decompose the formula into disjoint subformulas so that we can efficiently solve it in a divide-and-conquer manner. Specifically, the selection of the clause we branch on is guided by a bisection (i.e., a balanced separator) of a specific graph representation of the formula (the definitions will be given later). We note that this paradigm of utilizing graph bisections to guide branching decisions has been successfully applied to design fast polynomial-space algorithms for various problems (e.g., \cite{journals/talg/GaspersS17,DBLP:conf/cp/HoiJS20,DBLP:conf/walcom/Hoi0S024}).

\subsection{Dealing with $4^+$-clauses}
In this step, we pick up a $4^+$-clause $C$ and branch on: (B1) remove clause $C$; (B2) remove clause $C$ and set all literals in $C$ to $0$.

Two clauses $C, D \in \formula$ are called \emph{neighbors} if they share at least one common variable, i.e., $\var(C) \cap \var(D) \neq \emptyset$. For a clause $C \in \formula$, we let $N_{\formula}(C)$ denote the set of its neighbors, and define $N_{\formula}(C, 2) \coloneqq \{ D\mid D \in N_{\formula}(C), |D| = 2 \}$ to be the set of $2$-clauses adjacent to $C$.
\begin{lemma}\label{lem:measure-decrease-2occ-primal}
    Let $C$ be a clause in formula $\formula$. Let $\formula_1=\formula\setminus \{C\}$ and $\formula_2=\formula[C=0]\setminus\{C\}$.
    It holds that
    \begin{itemize}
        \item $m(\formula) - m(R(\formula_1))\ge \abs{C}+1$ and $m(\formula) - m(R(\formula_2))\ge 1$.
        \item $n(\formula) - n(R(\formula_1))\ge \abs{\var(C\cup N_{\formula}(C))}$ and $n(\formula) - n(R(\formula_2))\ge \abs{C}+\abs{N_{\formula}(C,2)}$.
    \end{itemize}
\end{lemma}
\begin{proof}
    We first note that the exhaustive application of the reduction rules introduces neither new clauses nor new variables. Thus, for any intermediate formula $\formula'$, it holds that $m(R(\formula')) \le m(\formula')$ and $n(R(\formula')) \le n(\formula')$. 

    In the first branch, let $\formula_1 = \formula \setminus \{C\}$. The explicit removal of $C$ decreases the number of clauses by $1$. Because $\formula$ is a reduced formula, every variable in $\var(C)$ is a $2$-variable and has exactly one other occurrence outside of $C$. By Property~3 of \cref{lem:reduced-formula}, any two clauses share at most one $2$-variable. Thus, the remaining $\abs{C}$ occurrences of the variables in $\var(C)$ belong to exactly $\abs{C}$ distinct clauses in $N_\formula(C)$. Upon the removal of $C$, all variables in $\var(C)$ become $1$-variables. According to \cref{rrule:1var}, each $1$-variable $x$ is assigned a value of $1$, which satisfies and removes its unique remaining clause in $N_\formula(C)$. Since these $\abs{C}$ clauses are mutually distinct, their removal yields $m(\formula) - m(R(\formula_1)) \ge \abs{C} + 1$. {Furthermore, by \cref{rrule:1var}, all other literals in the clause containing $x$ are set to $0$.} Since this process iterates over all $x \in \var(C)$, all variables in $\var(N_\formula(C))$ are assigned and eliminated. Therefore, $n(\formula) - n(R(\formula_1)) \ge \abs{\var(C \cup N_\formula(C))}$.

    In the second branch, let $\formula_2 = \formula[C=0] \setminus \{C\}$. The explicit removal of $C$ immediately yields $m(\formula) - m(R(\formula_2)) \ge 1$. For the number of variables, all literals in $C$ are assigned to $0$; thus, the $\abs{C}$ variables in $\var(C)$ are eliminated. We then consider the clauses in $N_\formula(C, 2)$. For each $2$-clause $D \in N_\formula(C, 2)$, let $\var(D) = \{x, y\}$ where $x \in \var(C)$ and $y \notin \var(C)$. The assignment to $x$ either satisfies $D$ or falsifies the literal of $x$ in $D$. If $D$ is satisfied, it is removed; the remaining $2$-variable $y$ loses an occurrence, becomes a $1$-variable, and is subsequently assigned and eliminated via \cref{rrule:1var}. If the literal of $x$ is falsified, $D$ shrinks to a $1$-clause, forcing $y$ to be assigned and eliminated via \cref{rrule:1cls}. In either case, the other variable $y$ in $D$ is eliminated. This guarantees an additional decrease of at least $\abs{N_\formula(C, 2)}$ variables. Consequently, $n(\formula) - n(R(\formula_2)) \ge \abs{C} + \abs{N_\formula(C, 2)}$.
\end{proof}

\begin{lemma}\label{lem:2occ-4cls-branching-m}
    Let $C$ be a clause with $\abs{C} \ge 4$ in a reduced formula $\formula$. Branching on $C$ yields a branching factor of at most $\tau(5, 1)<1.3248$ with respect to the number of clauses.
\end{lemma}
\begin{proof}
    {
    Let $\formula_1 = \formula \setminus \{C\}$ and $\formula_2 = \formula[C=0] \setminus \{C\}$. Let $\Delta_{m, 1}$ and $\Delta_{m, 2}$ be the decreases in the number of clauses in the respective branches. By \cref{lem:measure-decrease-2occ-primal}, we have $\Delta_{m, 1} \ge \abs{C} + 1\ge 5$ and $\Delta_{m, 2} \ge 1$, which yields a branching factor of at most $\tau(5, 1)<1.3248$.
    }
\end{proof}

\begin{lemma}\label{lem:2occ-4cls-branching-n}
    Let $C$ be a clause with $\abs{C} \ge 4$ in a reduced formula $\formula$. Branching on $C$ yields a branching factor of at most $\tau(9, 4)<1.1193$ with respect to the number of variables.
\end{lemma}
\begin{proof}
    Let $\formula_1 = \formula \setminus \{C\}$ and $\formula_2 = \formula[C=0] \setminus \{C\}$. Let $\Delta_{n, 1}$ and $\Delta_{n, 2}$ be the decreases in the number of variables in the respective branches. By \cref{lem:measure-decrease-2occ-primal}, we have $\Delta_{n, 1} \ge \abs{\var(C \cup N_\formula(C))}$ and $\Delta_{n, 2} \ge \abs{C} + \abs{N_\formula(C, 2)}$.  

    Let $E = \var(N_\formula(C)) \setminus \var(C)$ be the set of external variables, and $c_2 = \abs{N_\formula(C, 2)}$. Since $\abs{C} \ge 4$, {we have $\Delta_{n, 2} \ge \abs{C} + c_2 \ge 4$ and $\Delta_{n, 1} \ge \abs{C} + \abs{E} \ge 4$.}
    To get the desired branching factor $\tau(9, 4)$ via \cref{lem:vec-dominate}, it suffices to show that their sum $\Delta_{n, 1} + \Delta_{n, 2} \ge 13$. Since $\Delta_{n, 1} + \Delta_{n, 2} \ge 2\abs{C} + c_2 + \abs{E} \ge 8 + c_2 + \abs{E}$, it remains to prove $c_2 + \abs{E} \ge 5$.

    Assume for contradiction that $c_2 + \abs{E} \le 4$. We partition $E$ into two sets based on their occurrences in $N_\formula(C)$: $k_1$ variables appearing exactly once, and $k_2$ variables appearing exactly twice. Thus, $\abs{E} = k_1 + k_2 \le 4$. 
    Counting the literal occurrences in $N_\formula(C)$, the $\abs{C}$ clauses require at least $2c_2 + 3(\abs{C} - c_2) = 3\abs{C} - c_2$ occurrences. These are supplied by $\var(C)$ (exactly $\abs{C}$ occurrences) and $E$ (exactly $k_1 + 2k_2$ occurrences). This yields $ 3\abs{C} - c_2 \le \abs{C} + k_1 + 2k_2$, which simplifies to $ 2\abs{C} \le c_2 + k_1 + 2k_2$.
    
    {Recall that we assume for contraidction that $c_2 + \abs{E} = c_2 + k_1+k_2 \le 4$.}
    Because $\abs{C} \ge 4$, we have $8 \le 2\abs{C} \le (c_2 + k_1 + k_2) + k_2 \le 4 + k_2$, which implies $k_2 \ge 4$. 
    With $k_1 + k_2 = \abs{E} \le 4$, we have $k_2 = 4$, $k_1 = 0$, and $c_2 = 0$ (along with $\abs{C} = 4$). However, $k_1$ represents the number of variables shared between the induced subformula $C \cup N_\formula(C)$ and the rest of $\formula$. If $k_1 = 0$, this subformula, which contains exactly $\abs{C} + \abs{E} = 8 \le 10$ variables, shares $0$ variables with the rest of the formula. This contradicts Property 5 of \cref{lem:reduced-formula}, which states that any small subformula must share at least $2$ variables with the rest. 
    
    Therefore, $c_2 + \abs{E} \ge 5$. Consequently, $\Delta_{n, 1} + \Delta_{n, 2} \ge 13$. 
    With $\min(\Delta_{n,1}, \Delta_{n,2}) \ge 4$, the branching factor is at most $\tau(13-4, 4) = \tau(9, 4) < 1.1193$.
\end{proof}

\subsection{Solving 3-CNF instances via a Bisection Approach}
Now there are only $3$-clauses and $2$-clauses.
We adopt a graph-theoretic approach to design and analyze the algorithm. 
In what follows, we first introduce the specific graph representation used throughout the algorithm and then present the algorithm. 

\subparagraph*{Graph Representation.}
The \emph{dual graph} $G_{\formula}$ of a formula $\formula$ is the graph where each vertex corresponds to a clause in $\formula$, and two distinct vertices $C_i$ and $C_j$ are connected by an edge if and only if $\var(C_i)\cap \var(C_j)\neq \emptyset$. \cref{lem:reduced-formula} ensures that any pair of clauses shares at most one $2$-variable. Since every variable appears exactly twice, each variable corresponds to a unique edge between two clauses.

For any clause $C$, consider the effects of either removing it from the formula or assigning values to its variables. Let $D$ be a neighboring clause of $C$ and $x$ be the shared variable between $C$ and $D$.
If $C$ is removed, the shared variable $x$ becomes a 1-variable, triggering \cref{rrule:1var} to assign values to other variables in $D$ to satisfy and remove $D$. 
If the shared variable $x$ is assigned a value, $D$ is either satisfied or reduced in length by at least one; in particular, if $D$ is a 2-clause, shrinking to a 1-clause triggers \cref{rrule:1cls} to assign a value to its remaining variable.
Thus, if $D$ is a 2-clause, it is eliminated in both scenarios and affects its other neighbor in the same manner. Ultimately, this results in the removal of a chain of 2-clauses (i.e., a path of degree-2 vertices in $G_\formula$). 

Given the above chain reaction, we construct a \emph{multigraph} $G_{\formula}^*=(V^*, E^*)$ by smoothing out all degree-2 vertices in $G_\formula$, where
\begin{itemize}
    \item the vertex set $V^*$ consists of the $m_3(\formula)$ degree-3 vertices in $G_\formula$ (i.e., all 3-clauses in $\formula$);
    \item for two vertices $C_i, C_j \in V^*$, each edge between them in $E^*$ corresponds to a distinct connection in $G_\formula$, which is either a direct edge (i.e., sharing a variable) or a path of degree-2 vertices (i.e., a chain of 2-clauses).
\end{itemize}
{See \cref{fig:example-multigraph} for an illustrative example.} 

Note that any connected component in $G_{\formula}$ consisting entirely of degree-2 vertices is absent from $G_{\formula}^*$. This is harmless, as it induces an isolated 2-CNF formula, which can be solved in polynomial time by the following lemma.

\begin{figure}
    \centering
    \input{figures/example_graph}
    \caption{An illustrative example of the dual graph $G_{\formula}$ (left) and the multigraph $G_{\formula}^*$ (right) for $\formula$, which consists of four 3-clauses $C_1=(x_1 \lor \nl{x_4} \lor y_1), C_2=(\nl{x_1} \lor x_2 \lor y_2), C_3=(x_2 \lor x_3 \lor y_3), C_4=(x_3 \lor x_4 \lor y_4)$ and four 2-clauses $D_1=(\nl{y_1} \lor z_1), D_2=(z_1 \lor y_2), D_3=(y_3 \lor z_2), D_4=(\nl{z_2} \lor y_4)$. Each edge in $G_{\formula}$ corresponds to a single variable; in $G_{\formula}^*$, chains of 2-clauses are replaced with (potentially parallel) edges by smoothing out all degree-2 vertices $D_j$.}
    \label{fig:example-multigraph}
\end{figure}

\begin{lemma}\label{lem:2occ-2CNF-polytime}
    \ParityTwoOccSAT{} on $2$-CNF formulas can be solved in polynomial time.
\end{lemma}
\begin{proof}Given a $2$-CNF formula $\phi$, consider its dual graph $G_{\formula} = (V, E)$, where each vertex $v \in V$ corresponds to a clause in $\phi$. Two vertices are connected by an edge if their corresponding clauses share a common variable. Since each variable in $\phi$ occurs at most twice (2-occurrence) and each clause contains at most two literals, the maximum degree of $G$ is $\Delta(G) \le 2$, implying that $G$ is a disjoint union of {paths} and {cycles}. We process each connected component of $G$ as follows:
\begin{itemize}
    \item {Paths:} Let $P$ be a path in $G$. Consider an endpoint vertex of $P$. This vertex corresponds to a clause that either contains a literal whose variable appears only once in $\phi$, or is a unit clause. In either case, we can apply \cref{rrule:1cls} or \cref{rrule:1var} to eliminate the clause and the associated variable. 
    \item {Cycles:} For a cycle, we perform clause branching on an arbitrary node (i.e., clause) in it. This branching step generates two instances, each containing a path-structured component. Following the same reduction procedure for paths, the cycle is eliminated. 
\end{itemize}
Since each reduction rule and branching step can be executed in polynomial time, the total running time is polynomial. Thus, \ParityTwoOccSAT{} on 2-CNF formulas is polynomial-time solvable.
\end{proof}

Furthermore, $G_{\formula}^*$ might contain self-loops if a 3-clause $C = (x \lor y \lor z)$ connects to itself via a chain of 2-clauses with endpoints $y$ and $z$. To eliminate such loops, let $\formula_1$ be the subformula formed by $C$ and this chain, and $\formula_2 = \formula \setminus \formula_1$. Note that $\var(\formula_1) \cap \var(\formula_2) = \{x\}$, and both $\formula_1[x=0]$ and $\formula_1[x=1]$ are 2-CNF formulas. This allows us to apply \cref{rrule:small-subf-1}, with the only modification that $\parity(\formula_1[x=0])$ and $\parity(\formula_1[x=1])$ are computed in polynomial time via \cref{lem:2occ-2CNF-polytime}. Thus, there is no self-loop in $G_{\formula}^*$ for a reduced formula $\formula$.

\subparagraph*{Bisection.} 
A \emph{bisection} of a graph $G = (V, E)$ is a partition $(V_1, V_2)$ of the vertex set $V$ such that their sizes are as equal as possible, i.e., $V_1\cup V_2=V$, $V_1\cap V_2=\emptyset$ and $\left| |V_1| - |V_2| \right| \leq 1$. The \emph{size} of a bisection is the number of edges with one endpoint in $V_1$ and the other in $V_2$. The minimum size over all possible bisections of $G$ is called the \emph{bisection width} of the graph.
We rely on the following upper bound on the bisection width of cubic graphs.

\begin{theorem}[Monien and Preis~\cite{journals/jda/MonienP06}]\label{lem:bisection-width}
   For any $\varepsilon>0$, there exists an integer $n_{\varepsilon}$ such that for any cubic graph $G$ with $|V(G)| \geq n_{\varepsilon}$, the bisection width of $G$ is at most $(\frac{1}{6}+\varepsilon)|V(G)|$.
\end{theorem}
As noted in \cite{journals/ipl/FominH06,journals/talg/GaspersS17}, there is a polynomial-time algorithm (denoted by \texttt{Monien-Preis}$(G)$) that computes the corresponding bisection. 
We note that this result also applies to multigraphs {(see Appendix~\ref{appendix:bisection-multigraphs} for a proof)}.

\subparagraph*{The Algorithm.}
\input{Algs/alg-bisection}
Now we are ready to describe our algorithm, presented in \cref{alg:2occ-bisection}.
The algorithm follows a divide-and-conquer strategy where the branching process is guided by a partition initialized as a bisection of the multigraph $G_{\formula}^*$. Specifically, it maintains this disjoint partition $(A, B)$ of the $3$-clauses, aiming to eliminate variables shared between these sets (i.e., edges between them) through branching. The algorithm is initially called with $A$ being the set of all $3$-clauses and $B = \emptyset$. Thus, initially or whenever the partition becomes trivial (i.e., $B = \emptyset$), the \texttt{Monien-Preis} procedure is invoked on $G_{\formula}^*$ to compute a bisection of guaranteed size. If no edges exist between $A$ and $B$, the formula decomposes into two disjoint subformulas $\formula_A$ and $\formula_B$, and the algorithm recursively solves them as independent instances. Otherwise, the algorithm selects an edge connecting $A$ and $B$, and applies clause branching on an endpoint $3$-clause $C$ of this edge. Here, we alternate the selection of $C$ between $A$ and $B$ across recursion levels to better balance the number of $3$-clauses eliminated in both sides, which is crucial for bounding the overall worst-case running time.

\subparagraph*{The Analysis.}
Let $m_3(\formula)$ be the number of $3$-clauses in $\formula$. {Then we have $|V(G_{\formula}^*)|=m_3(\formula)$} {and $L(\formula)\ge 3m_3(\formula)$. Since we are considering \ParityTwoOccSAT{}, $L(\formula)\leq 2n(\formula)$.} Thus, $m_3(\formula)\leq \frac{2}{3}n(\formula)$. We will first obtain a running time bound in terms of $m_3(\formula)$ and then transform it to a running time bound in terms of $m(\formula)$ or $n(\formula)$.
\begin{lemma}\label{lem:2occ-3CNF-m3}
    \ParityTwoOccSAT{} on 3-CNF formulas can be solved in $\bigOstar{1.1487^{m_3}}$ time, where $m_3$ is the number of 3-clauses.
\end{lemma}
\begin{proof}
    Let us fix a sufficiently small constant $\varepsilon > 0$ (e.g., $\varepsilon = 10^{-9}$) and let $n_\varepsilon$ be as defined in \cref{lem:bisection-width}. If the formula $\formula$ has $m_3(\formula) \le n_\varepsilon$, we can solve it in polynomial time by exhaustively applying clause branching on this constant-sized set of $3$-clauses (in each branch, $m_3(\formula)$ decreases by at least one) and solving the resulting $2$-CNF instances by \cref{lem:2occ-2CNF-polytime}.

    % For $m_3(\formula) > n_\varepsilon$, we proceed as described in \cref{alg:2occ-bisection}. 
    {When $m_3(\varphi) > n_{\epsilon}$, the algorithm executes Line~\ref{line:bisection-step} onwards. Before delving into the detailed proof, we provide a high-level outline. We will first define a new measure $\rho$ that is initially bounded by $m_3(\formula)$. Thus, establishing a runtime bound of $O^*(\gamma^\rho)$ for some $\gamma>1$ yields a runtime bound of $O^*(\gamma^{m_3(\formula)})$. We then demonstrate that (i) the bisection step (Line~\ref{line:bisection-step}) does not increase the measure;  (ii) the divide-and-conquer step (Line~\ref{line:dc-step}) does not affect the exponential part of the time complexity; and (iii) the branching step yields a branching factor of at most $1.1487$.}

    We adopt the following measure for any valid input $(\formula, A, B)$ of the algorithm:
    \[
        \rho(\formula, A, B) = \max(\abs{A}, \abs{B}) + (3 - \varepsilon')\abs{S},
    \]
    where $\abs{A}$ and $\abs{B}$ are the number of $3$-clauses in the respective partitions, and $\abs{S}$ is the number of edges between $A$ and $B$. We choose $\varepsilon'$ 
    such that $(3 - \varepsilon')(\frac{1}{6}+\varepsilon) \le \frac{1}{2}-\frac{1}{n_{\epsilon}}<\frac{1}{2}-\frac{1}{m_3(\formula)}$. 
    Crucially, since the algorithm is initially invoked with $B = \emptyset$ and $A = V(G_{\formula}^*)$, the initial measure is exactly $\rho(\formula, A, \emptyset) = \abs{A} = m_3(\formula)$. {Furthermore, all reduction rules do not introduce new variables or clauses; thus, applying them does not increase the measure $\rho$.}
    
    {We first show that the bisection step does not increase the measure.} When $B=\emptyset$, the algorithm applies the \texttt{Monien-Preis} procedure on $G_{\formula}^*$, generating a new partition $(A', B')$ and the corresponding $S'$ with $\abs{A'}, \abs{B'} \le \frac{\abs{A}}{2}+1$ and $\abs{S'} \le (\frac{1}{6} + \varepsilon)\abs{A}$. Since $(3 - \varepsilon')(\frac{1}{6} + \varepsilon) \le \frac{1}{2}-\frac{1}{m_3(\formula)}=\frac{1}{2}-\frac{1}{\abs{A}}$, the new measure $\rho(\formula, A', B')$ is bounded by:
    \[
        \rho(\formula, A', B') \le \left(\frac{\abs{A}}{2}+1\right) + (3-\varepsilon')\left(\frac{1}{6} + \varepsilon\right)\abs{A} \le \frac{\abs{A}}{2}+1 + \frac{\abs{A}}{2}-1 = \abs{A} = \rho(\formula, A, B).
    \]

    Let $T(\rho)$ be the maximum number of leaves in the search tree generated by the algorithm on a state with measure $\rho$, and let $\gamma>1$ be the largest branching factor generated by the subsequent branching step. {We will prove that $T(\rho) \in \bigOstar{\gamma^{\rho}}$.}
    
    {We show that the divide-and-conquer step does not affect the exponential part of the time complexity in terms of our measure $\rho$.}
    To this end, we prove by induction on the measure $\rho$ that $T(\rho) \le m_3(\formula) \gamma^\rho\in \bigOstar{\gamma^{\rho}}$.
    In this step, $S = \emptyset$, and the formula decomposes into disjoint subformulas $\formula_A$ and $\formula_B$. Let $\rho_A := \rho(\formula_A, A, \emptyset)$ and $\rho_B := \rho(\formula_B, B, \emptyset)$ be the measures of the two subproblems, respectively. Since $\rho = \max(\abs{A}, \abs{B})$, we have $\rho_A = \abs{A} \le \rho$ and $\rho_B = \abs{B} \le \rho$. By the inductive hypothesis, the total cost of solving these two subformulas independently is:
    \[ 
    T(\rho) \le T(\rho_A) + T(\rho_B) \le \abs{A}\gamma^{\rho_A} + \abs{B}\gamma^{\rho_B} \le (\abs{A} + \abs{B})\gamma^\rho = m_3(\formula) \gamma^\rho. 
    \]
    
    {It remains to show that the branching step generates a branching factor of at most $\gamma<1.1487$.
    As established in the paragraph ``Graph Representation'', when a clause is removed from the formula or its variables are assigned, the operation triggers a chain reaction that propagates through its shared variables. In the multigraph $G_\formula^*$, this corresponds to the deletion of the corresponding vertex and the removal of all its incident edges.}
    
    When $S \neq \emptyset$, the algorithm selects an edge in $S$ connecting $C_A \in A$ and $C_B \in B$. Assume that we branch on $C_A$. Let $\formula_1=\formula\setminus\{C_A\}$ and $\formula_2=\formula[C_A=0]\setminus\{C_A\}$. For each branch $i \in \{1, 2\}$, let $A_{i}$ and $B_{i}$ be the subsets of clauses that remain as $3$-clauses in $R(\formula_i)$. To formalize the measure drop, let $\Delta_{i,A} = \abs{A} - \abs{A_i}$, $\Delta_{i,B} = \abs{B} - \abs{B_i}$, and $\Delta_{i,S} = \abs{S} - \abs{S_i}$ denote the reductions in the respective parts. We show that for each branch $i\in\{1,2\}$, either $\Delta_{i,S} \ge 2$, or $\Delta_{i,S} = 1$, $\Delta_{i,A}\ge 3$ and $\Delta_{i,B}\ge 1$.

If $C_A$ is incident to two or more edges in $S$, the chain reaction in both $\formula_1$ and $\formula_2$ removes at least two edges from $S$, guaranteeing $\Delta_{i,S} \ge 2$ for $i \in \{1, 2\}$. 
% This already decreases the measure $\rho$ by at least $2(3-\varepsilon') = 6-2\varepsilon'$ in both branches, yielding a branching factor of at most $\tau(6-2\varepsilon',6-2\varepsilon') < 1.1225$, which is sufficient. Since the branching factor is good enough if we have $\Delta_{i,S} \ge 2$ for $i\in \{1,2\}$, in what follows, we can further assume that $\Delta_{i,S} \le 1$ for $i\in \{1,2\}$. In particular, $C_A$ and $C_B$ are each incident to exactly one edge in $S$.

If $C_A$ is incident to exactly one edge in $S$, since $C_A$ is a $3$-clause, its two remaining incident edges must connect to other clauses within $A$. If they connect to distinct clauses, the chain reaction in either $\formula_i$ propagates along both edges, satisfying or shrinking at least two additional clauses in $A$, yielding $\Delta_{i,A} \ge 3$. Alternatively, if the two internal edges connect to the same clause $C'_A$, the chain reaction either removes $C'_A$ or shrinks $C'_A$ to a $1$-clause. In the latter case, the $1$-clause will also be removed. Because $C'_A$ is a $3$-clause, its removal triggers a further chain reaction along its third incident edge and removes or shrinks a third $3$-clause. If this third $3$-clause belongs to $A$, we have $\Delta_{i,A} \ge 3$; otherwise, we remove one more edge in $S$ and have $\Delta_{i,S} \ge 2$. Furthermore, the chain reaction propagates along the single edge in $S$ connecting $C_A$ to $C_B$. This eliminates the edge in $S$, yielding $\Delta_{i,S} \ge 1$, and eventually satisfies or shrinks $C_B$, ensuring $\Delta_{i,B} \ge 1$. Consequently, for both $i \in \{1, 2\}$, we guarantee that either $\Delta_{i,S} \ge 2$, or $\Delta_{i,S} = 1$, $\Delta_{i,A}\ge 3$ and $\Delta_{i,B}\ge 1$.

Recall that the algorithm employs a two-level alternating branching strategy. It first branches on $C_A \in A$ to obtain $\formula_1$ and $\formula_2$. Then, in each resulting subformula $\formula_i$, it selects a clause in $B_i$ incident to $S$ and branches on it to generate subformulas $\formula_{i,1}$ and $\formula_{i,2}$. Let $\Delta_{i,j,A}$, $\Delta_{i,j,B}$, and $\Delta_{i,j,S}$ denote the respective component reductions in this second step for $j \in \{1, 2\}$. By the same argument as above, the second level guarantees that either $\Delta_{i,j,S} \ge 2$, or $\Delta_{i,j,S} \ge 1$, $\Delta_{i,j,B} \ge 3$, and $\Delta_{i,j,A} \ge 1$. Combining the two levels across any of the four branches in the search tree, it holds for the cumulative reduction for $A$, $B$, $S$ that either (a) $\Delta_{i,S} + \Delta_{i,j,S} \ge 4$; or (b) $\Delta_{i,A} + \Delta_{i,j,A} \ge 1$, $\Delta_{i,B} + \Delta_{i,j,B} \ge 1$, and $\Delta_{i,S} + \Delta_{i,j,S} \ge 3$; or (c) $\Delta_{i,A} + \Delta_{i,j,A} \ge 4$, $\Delta_{i,B} + \Delta_{i,j,B} \ge 4$, and $\Delta_{i,S} + \Delta_{i,j,S} \ge 2$.
The total measure drop in terms of $\rho$ is at least $\min\{4(3-\varepsilon'), 1+3(3-\varepsilon'), 4 + 2(3-\varepsilon')\} = 10 - 3\varepsilon'$. This uniform drop across all four branches yields a branching factor of at most $\tau(10-3\varepsilon', 10-3\varepsilon', 10-3\varepsilon', 10-3\varepsilon')<1.1487$.
\end{proof}

Since $m_3(\formula)\le m(\formula)$ and $m_3(\formula)\le \frac{2}{3}n(\formula)$, by \cref{lem:2occ-3CNF-m3} we have
\begin{lemma}\label{lem:2occ-3CNF-m-n}
    \ParityTwoOccSAT{} on 3-CNFs with $m$ clauses and $n$ variables can be solved using polynomial space in $\bigOstar{1.1487^m}$ time or $\bigOstar{1.0969^n}$ time.
\end{lemma}

\subsection{The Final Results}
By \cref{lem:2occ-4cls-branching-m}, branching on $4^+$-clauses yields a branching factor of at most $\tau(5, 1)<1.3248$ with respect to the number of clauses $m$, which dominates the $\bigO{1.1487^m}$ time bound for solving the resulting $3$-CNF instances given by \cref{lem:2occ-3CNF-m-n}. Thus, we have
\begin{theorem}\label{thm:result-2occ-m}
    \ParityTwoOccSAT{} with $m$ clauses can be solved in $\bigOstar{\ResultTwoOccM}$ time and polynomial space.
\end{theorem}
Similarly, the branching factor of $\tau(9, 4)<1.1193$ in terms of the number of variables $n$ from \cref{lem:2occ-4cls-branching-n} dominates the $\bigO{1.0969^n}$ bound for $3$-CNF instances. Thus, we obtain 
\begin{theorem}\label{thm:result-2occ-n}
    \ParityTwoOccSAT{} with $n$ variables can be solved in $\bigO{\ResultTwoOccN}$ time and polynomial space.
\end{theorem}

%% file: figures/example_graph.tex
\begin{tikzpicture}[
    % ==========================================
    % 1. Parameter Settings
    % ==========================================
    /utils/exec={
        \def\nSize{6mm}      
        \def\hDist{2.8}      % Slightly wider to prevent crowding of long path labels
        \def\vDist{1.0}      % Increased height to provide space for internal labels
        \def\chainOff{0.8}   
        \def\graphGap{6.0}   
        \def\labelDist{0.2cm}
    },
    % ==========================================
    % 2. Style Definitions
    % ==========================================
    clause_node/.style={circle, draw=black, thick, fill=white, minimum size=\nSize, font=\scriptsize\bfseries, inner sep=0pt},
    % Variable/sequence labels: uniform gray, no background fill
    var_label/.style={text=black!50, font=\fontsize{7}{8}\selectfont, inner sep=2pt, midway},
    all_edge/.style={-, thick, color=black}
]

% =====================================================================
% Left side: Dual Graph G_\formula
% =====================================================================
\begin{scope}[local bounding box=G_left]
    \node[clause_node] (C1) at (0, \vDist) {$C_1$};
    \node[clause_node] (C2) at (\hDist, \vDist) {$C_2$};
    \node[clause_node] (C3) at (\hDist, 0) {$C_3$};
    \node[clause_node] (C4) at (0, 0) {$C_4$};

    \node[clause_node] (D1) at (0.3*\hDist, \vDist+\chainOff) {$D_1$};
    \node[clause_node] (D2) at (0.7*\hDist, \vDist+\chainOff) {$D_2$};
    \node[clause_node] (D3) at (0.7*\hDist, -\chainOff) {$D_3$};
    \node[clause_node] (D4) at (0.3*\hDist, -\chainOff) {$D_4$};

    % Direct connections
    \draw[all_edge] (C1) -- node[var_label, above] {$x_1$} (C2);
    \draw[all_edge] (C2) -- node[var_label, right] {$x_2$} (C3);
    \draw[all_edge] (C3) -- node[var_label, below] {$x_3$} (C4);
    \draw[all_edge] (C4) -- node[var_label, left]  {$x_4$} (C1);

    % Chain connections
    \draw[all_edge] (C1) -- node[var_label, left=1pt]  {$y_1$} (D1);
    \draw[all_edge] (D1) -- node[var_label, above]     {$z_1$} (D2);
    \draw[all_edge] (D2) -- node[var_label, right=1pt] {$y_2$} (C2);

    \draw[all_edge] (C3) -- node[var_label, right=1pt] {$y_3$} (D3);
    \draw[all_edge] (D3) -- node[var_label, below]     {$z_2$} (D4);
    \draw[all_edge] (D4) -- node[var_label, left=1pt]  {$y_4$} (C4);
\end{scope}

% =====================================================================
% Right side: Multigraph G_\formula^*
% =====================================================================
\begin{scope}[xshift=\graphGap cm, local bounding box=G_right]
    \node[clause_node] (C1s) at (0, \vDist) {$C_1$};
    \node[clause_node] (C2s) at (\hDist, \vDist) {$C_2$};
    \node[clause_node] (C3s) at (\hDist, 0) {$C_3$};
    \node[clause_node] (C4s) at (0, 0) {$C_4$};

    % 1. Horizontal direct edges (Middle lines)
    \draw[all_edge] (C1s) -- node[var_label, above] {$x_1$} (C2s);
    \draw[all_edge] (C4s) -- node[var_label, below] {$x_3$} (C3s);
    
    % 2. Vertical direct edges
    \draw[all_edge] (C2s) -- node[var_label, right] {$x_2$} (C3s);
    \draw[all_edge] (C4s) -- node[var_label, left]  {$x_4$} (C1s);

    % 3. Compressed path edges (Outer bent edges)
    % Upper arc: Label placed above, pointing towards the rectangle interior
    \draw[all_edge] (C1s) to[bend left=35] 
        node[var_label, above=2pt] {$D_1,D_2$} (C2s);

    % Lower arc: Label placed below, pointing towards the rectangle interior
    \draw[all_edge] (C4s) to[bend right=35] 
        node[var_label, below=2pt] {$D_4,D_3$} (C3s);
\end{scope}

% =====================================================================
% Transition arrow and alignment labels
% =====================================================================
\draw[-{Stealth[length=4mm]}, line width=2pt, color=black!10] 
    (\hDist+0.8, \vDist*0.5) -- (\graphGap-0.8, \vDist*0.5);

\node[anchor=north, font=\small\bfseries] at ([yshift=-\labelDist] G_left.center |- G_left.south) {Dual Graph $G_{\formula}$};
\node[anchor=north, font=\small\bfseries] at ([yshift=-\labelDist] G_right.center |- G_left.south) {Multigraph $G_{\formula}^*$};

\end{tikzpicture}

%% file: Algs/alg-bisection.tex
\renewcommand{\algorithmicrequire}{\textbf{Input:}}
\renewcommand{\algorithmicensure}{\textbf{Output:}}
\renewcommand{\algorithmiccomment}[1]{\quad  \textcolor{gray}{$\triangleright$ #1}}
\begin{algorithm}[t]
\caption{\texttt{Bisection-Solve}($\formula, A, B$)}\label{alg:2occ-bisection}
\begin{algorithmic}[1]
\Require A reduced formula $\formula$, and a disjoint partition $(A, B)$ of the $3$-clauses in $\formula$.
\Ensure The parity of the number of satisfying assignments of $\formula$ ($1$ for odd, $0$ for even).

\LeftComment{Let $\varepsilon>0$ be a fixed small constant and $n_{\varepsilon}$ be the integer defined in \cref{lem:bisection-width}.}
\If{{$|V(G_{\formula}^*)|\leq n_{\varepsilon}$}}
\State Solve the instance by iteratively branching on $3$-clauses and solving the resulting $2$-CNF instances using \cref{lem:2occ-2CNF-polytime}, and \Return the result.
\EndIf

\If{$B = \emptyset$} \mycomment{bisection step}\label{line:bisection-step}
    \State $(A',B')\gets \texttt{Monien-Preis}(G_\formula^*)$;
    \State \Return \texttt{Bisection-Solve}($\formula, A', B'$);
\EndIf

\If{there are no edges between $A$ and $B$} \mycomment{divide-and-conquer step}\label{line:dc-step}
    \State Let $\formula_A$ and $\formula_B$ be the disjoint subformulas induced by $A$ and $B$, respectively.
    \State \Return \texttt{Bisection-Solve}($\formula_A, A, \emptyset$) $\land$ \texttt{Bisection-Solve}($\formula_B, B, \emptyset$).
\EndIf

\LeftComment{branching step}
\State Select an edge between $A$ and $B$, and pick its endpoint $3$-clause $C \in A \cup B$, alternating the selection of $C$ between $A$ and $B$ at each successive level of the recursion. 
\State Let $\formula_1 = \formula \setminus \{C\}$ and $\formula_2 = \formula[C=0] \setminus \{C\}$.
\For{$i \in \{1, 2\}$}
    \State Exhaustively apply reduction rules on $\formula_i$.
    \State Let $A_i \subseteq A$ and $B_i \subseteq B$ be the subsets of clauses that remain as $3$-clauses in $\formula_i$.
\EndFor
\State \Return \texttt{Bisection-Solve}($\formula_1, A_1, B_1$) $\oplus$ \texttt{Bisection-Solve}($\formula_2, A_2, B_2$).
\end{algorithmic}
\end{algorithm}

%% file: sections/Parity-SAT-L.tex
\section{An Algorithm for Parity-SAT in Terms of $L$}\label{sec:L}
In this section, we present an algorithm for the general \ParitySAT{} problem that runs in $\bigO{1.1052^L}$ time and polynomial space.
Our algorithm processes variables in decreasing order of their degrees. By iteratively branching on and eliminating variables of high degrees, the algorithm ultimately reduces the given instance to \ParityTwoOccSAT{}. At this stage, we invoke the algorithm developed in the \cref{sec:2-occ} to solve the remaining part. 
In our analysis, we employ the measure-and-conquer method~\cite{journals/jacm/FominGK09}.
{In what follows, we first introduce the measure used in our analysis, and then present what each step of the algorithm handles and its complexity.}

\subsection{The Measure}\label{subsec:L-measure}
{For a formula $\formula$, let $n_i(\formula)$ be the number of $i$-variables.} We adopt the following measure:
\begin{align*}
    \mu(\formula):= \sum_{i\geq 1} w_i\cdot n_i(\formula), \text{where } w_1=0, w_2=1.5, \text{and } w_i=i\text{~for~$i\geq 3$}.
\end{align*}
Since $w_i \le i$ for all $i \ge 1$, it holds that
$
    \mu(\formula)\leq \sum_{i\ge 1} i\cdot n_i(\formula)\leq L(\formula).
$
Thus, a running-time bound of $\bigOstar{c^{\mu(\formula)}}$ for some $c>1$ implies a running-time bound of $\bigOstar{c^{L(\formula)}}$. 

Intuitively, since our algorithm processes variables in decreasing order of their degrees and relies on an efficient subroutine for \ParityTwoOccSAT{}, this measure allows us to smoothly propagate the efficiency from the low-degree base case up to higher degrees. 

For every integer $i\geq 1$, let $\delta_i \coloneqq w_i - w_{i-1}$ denote the marginal weight, i.e., the decrease in the measure when a variable's degree drops from $i$ to $i-1$. 
By the definition of the weights, we have $\delta_2 = \delta_3 = 1.5$ and $\delta_i = 1$ for all $i \ge 4$. Furthermore, the following properties hold:
\begin{equation*}
    w_3 = 2w_2, \quad w_2 = \delta_3, \quad \text{and} \quad \delta_i \leq \delta_{i-1} \text{ for every } i \geq 3.
\end{equation*}

\begin{lemma}\label{lem:mu-reduction-measure}
    For any formula $\formula$, applying any reduction rule does not increase $\mu(\formula)$.
\end{lemma}
\begin{proof}
    Except for \cref{rrule:two-2cls-complementary}, all other reduction rules only involve removing clauses, deleting literal occurrences, or assigning values to variables. These rules never introduce new variables or literal occurrences. Consequently, the degree of any remaining variable does not increase. By the monotonicity of $w_i$, these rules cannot increase $\mu(\formula)$.

    \cref{rrule:two-2cls-complementary} sets $x = \nl{y}$ for two $2$-clauses $(x \lor y)$ and $(\nl{x} \lor \nl{y})$, merging two variables into one. Let $a$ and $b$ be the original degrees of $x$ and $y$, respectively. Since $x$ and $y$ co-occur in these two clauses, we have $a \ge 2$ and $b \ge 2$. After the replacement, both clauses become tautologies and are immediately removed by \cref{rrule:taut}. Thus, the degree $d'$ of the new variable is bounded by $d' \le a + b - 4$. It can be verified that $w_{d'}\le w_{a}+w_{b}$ for all $d'\le a+b-4$ with $a,b\ge 2$. Thus, the application of \cref{rrule:two-2cls-complementary} does not increase the measure.
\end{proof}

\subsection{The Algorithm}\label{subsec:alg-L}
\begin{table}[htpb]
    \centering
    \input{tables/result-L-overview}
    \caption{Overview of branching factors for each step in the algorithm.}\label{tab:result-L-overview}
\end{table}

Our algorithm consists of six steps. An overview of the branching objects, vectors, and the resulting factors for each step is summarized in \cref{tab:result-L-overview}. 

In the algorithm, Step 1 first applies simple branching to eliminate all $4^+$-variables. 
Next, Steps 2--5 deal with $3$-variables, which is the core of the algorithm, by progressively purifying their local structures. 
Specifically, Step 2 applies clause branching (\cref{lem:clause-branching}) to eliminate $4^+$-clauses. Step 3 handles variables with both positive and negative occurrences, applying clause branching if the negative literal appears in a $3$-clause (Step 3.1) and simple branching if it is in a $2$-clause (Step 3.2). Step 4 applies simple branching to $3$-variables contained in $2$-clauses. 
Finally, Step 5 resolves the remaining well-structured cases, where all $3$-variables are positive and appear only in $3$-clauses. Step 5.1 applies simple branching to any {$3$-variable} whose clauses either contain a $2$-variable or share another $3$-variable. For the remaining variables, whose three clauses contain only distinct $3$-variables, Step 5.2 applies a tailored variable branching (\cref{lem:variable-branching}). 
After eliminating all $3$-variables, Step 6 reduces the instance to \ParityTwoOccSAT{}. At this point, the remaining formula consists entirely of $2$-variables, meaning $\mu(\formula) = w_2n(\formula) =1.5n(\formula)$. Consequently, it can be solved in $\bigOstar{1.1193^{\frac{\mu(\formula)}{1.5}}}\subseteq \bigOstar{1.0781^{\mu(\formula)}}$ time. 

The bottleneck of the algorithm arises in Step 5.1, which yields a worst-case branching factor of $1.1052$. 

At a high level, our analysis quantifies the measure decrease in each branch. We trace which variables are assigned or experience a degree drop to compute the measure decrease. Recall that dealing with $3$-variables constitutes the core part of our algorithm. Setting $\delta_3 = \delta_2 = w_2 = 1.5$ significantly simplifies the analysis in these steps: in many cases, we do not need to distinguish whether a variable's degree drops from $3$ to $2$, or from $2$ to $1$, because both cases yield the same measure decrease of $w_2=1.5$.
We primarily exploit the cascading measure drops from two specific scenarios: (1) when a clause shrinks into a $1$-clause, triggering \cref{rrule:1cls} to assign its variable; and (2) when a variable becomes a $1$-variable (e.g., when a $2$-variable loses one occurrence due to the removal of clause containing it, or a $3$-variable loses two occurrences simultaneously), triggering \cref{rrule:1var} to force assignments on variables in its remaining clause. Other reduction rules are mainly used to guarantee certain structural properties of the reduced formula, enabling us to deduce tighter bounds.

\subsubsection{Step 1: Dealing with $4^+$-variables} 

We first apply simple branching to eliminate $4^+$-variables if any.

\begin{lemma}\label{lem:L-step1-branching}
    Let $x$ be a variable in $\formula$ with the highest degree $d\geq 4$. Branching on (B1) $x=1$; (B2) $x=0$ generates a branching factor of at most $ \tau(12,4)<1.1003$.
\end{lemma}
\begin{proof}
    Let $\Delta_1 = \mu(\formula) - \mu(R(\formula[x=1]))$ and $\Delta_0 = \mu(\formula) - \mu(R(\formula[x=0]))$ be the measure decrease in these two branches. Note that $\Delta_1, \Delta_0 \ge w_d$ since $x$ is assigned in both branches.
    To prove the lemma, it suffices to show $\Delta_1 + \Delta_0\ge 2w_d + 2d\delta_d$ since then 
    the worst case is $\tau(2w_d + 2d\delta_d - w_d, w_d)=\tau(w_d + 2d\delta_d, w_d) = \tau(3d, d)\le \tau(12, 4) < 1.1003$.
    
    Let $\mathcal{C}_x$ be the set of the $d$ clauses containing $x$ or $\nl{x}$. We partition $\mathcal{C}_x$ into $d_2$ clauses of length $2$, and $d_{\ge 3}$ clauses of length at least $3$.
    For every variable $y \in \var(\mathcal{C}_x) \setminus \{x\}$, let $k(y)$ be the number of clauses in $\mathcal{C}_x$ containing $y$ or $\nl{y}$, and let $t(y)$ be the number of $2$-clauses in $\mathcal{C}_x$ containing $y$ or $\nl{y}$. We have
    \[
        \sum_{y \in \var(\mathcal{C}_x) \setminus \{x\}} (k(y) + t(y)) \ge (2d_{\ge 3} + d_2 + d_2) = 2d
    \]
    since each $3^+$-clause contributes at least $2$ to $\sum k(y)$, and each $2$-clause contributes $1$ to $\sum k(y)$ and $1$ to $\sum t(y)$.
    
    To avoid double-counting, we analyze the measure decrease contributed by each individual variable. The variable $x$ is assigned a value in both branches, so its total contribution to $\Delta_1 + \Delta_0$ is exactly $2w_d$. For every other variable $y \in \var(\mathcal{C}_x) \setminus \{x\}$, let $Q_1(y)$ and $Q_0(y)$ be its measure drop in branch $x=1$ and $x=0$ respectively, and let its total contribution be $Q(y) = Q_1(y) + Q_0(y)$.
    Our goal is to prove that for every $y \in \var(\mathcal{C}_x) \setminus \{x\}$, its contribution satisfies $Q(y) \ge (k(y) + t(y))\delta_d$. If this holds, we obtain the desired bound:
    \[
        \Delta_1 + \Delta_0 \ge 2w_d + \sum_{y \in \var(\mathcal{C}_x) \setminus \{x\}} Q(y)\geq 2w_d+\sum_{y \in \var(\mathcal{C}_x) \setminus \{x\}} (k(y) + t(y))\delta_d\ge 2w_d+ 2d\delta_d.
    \]

    It remains to prove $Q(y) \ge (k(y) + t(y))\delta_d$ for every $y \in \var(\mathcal{C}_x) \setminus \{x\}$.
    Recall that $w_2 = 1.5$ and $w_i = i$ for $i \ge 3$. Because variables of degree $1$ are eliminated, we have $w_1 = 0$, giving marginal weights $\delta_2 = w_2 - w_1 = 1.5$ and $\delta_i = w_i - w_{i-1} \ge 1$ for $i \ge 3$. Thus, for any variable $y$, losing one occurrence drops the measure by at least $\delta_{\deg(y)}$, and losing $k$ occurrences reduces its measure by at least $k \cdot \delta_d$ (noting that $\delta_i \ge \delta_d = 1$ since $d \ge 4$). 
    
    Let $k_1(y)$ and $k_0(y)$ be the number of clauses in $\mathcal{C}_x$ containing $y$ that are satisfied by $x=1$ and $x=0$, respectively. We prove the lower bound for $Q(y)$ case by case:
    \begin{itemize}
        \item \textbf{Case 1: $t(y) = 0$.} We need $Q(y) \ge k(y)\delta_d$.
        \begin{itemize}
            \item If $\deg(y) \ge 3$: $y$ loses at least $k_1(y)$ occurrences in branch $x=1$ and $k_0(y)$ in $x=0$. Clearly, $Q(y) = Q_1(y) + Q_0(y)\ge k_1(y)\delta_d + k_0(y)\delta_d = k(y)\delta_d$.
            \item If $\deg(y) = 2$: We have $k(y) \in \{1, 2\}$. If $k(y) = 1$, exactly one branch satisfies the clause containing $y$, so $Q(y) \ge w_2 \ge \delta_d = k(y)\delta_d$. If $k(y) = 2$, both literals of $y$ appear in $\mathcal{C}_x$. If $k_1(y)=2$, every clause containing $y$ would also contain literal $x$, meaning $x$ dominates $y$. This implies \cref{rrule:dominate} is applicable, contradicting that $\formula$ is reduced. Thus, $k_1(y)\le 1$. Similarly, $k_0(y)\le 1$. Since $k(y)=k_0(y)+k_1(y)=2$, we have $k_1(y) = k_0(y) = 1$ and $Q_1(y), Q_0(y) \ge w_2$, yielding $Q(y) \ge 2w_2 = 3 \ge 2\delta_d = k(y)\delta_d$.
        \end{itemize}
        
        \item \textbf{Case 2: $t(y) \ge 1$.} By \cref{lem:reduced-formula}, any two $2$-clauses have at most one common variable. Thus, $t(y)\leq 1$. Consequently, $t(y)=1$ and we need $Q(y) \ge (k(y) + 1)\delta_d$. 
        
        Let $C \in \mathcal{C}_x$ be the unique $2$-clause containing $y$. Without loss of generality, assume the literal of $x$ in $C$ is falsified in branch $x=0$. Consequently, in branch $x=0$, $C$ becomes a $1$-clause, which triggers \cref{rrule:1cls} and completely eliminates $y$, yielding a measure drop $Q_0(y) \ge w_{\deg(y)}$. In branch $x=1$, $C$ is satisfied, so $y$ loses at least this occurrence; thus, $Q_1(y) \ge \delta_{\deg(y)}$. We analyze the total drop $Q(y)$ based on $\deg(y)$:
        \begin{itemize}
            \item If $\deg(y) \ge 3$: We have $Q_0(y) \ge w_{\deg(y)} \ge \deg(y)\delta_d \ge k(y)\delta_d$ and $Q_1(y) \ge \delta_{\deg(y)} \ge \delta_d$, which gives $Q(y) = Q_0(y) + Q_1(y) \ge (k(y)+1)\delta_d$. 
            \item If $\deg(y) = 2$: We have $Q_0(y) \geq w_2$ and $Q_1(y) \ge \delta_2 = w_2$. Since $k(y) \le \deg(y) = 2$, we obtain $Q(y) = Q_0(y) + Q_1(y) \ge 2w_2 = 3 \ge (k(y) + 1)\delta_d$.
        \end{itemize}
    \end{itemize}
    This completes the proof.
\end{proof}

\subsubsection{Step 2--5: Dealing with $3$-variables} 
We will apply several branching rules to deal with $3$-variables.
Our branching operations frequently involve the removal of clauses. 
We first analyze the measure decrease caused by this process when a single clause is removed.

\begin{lemma}\label{lem:measure-drop-remove-clause}
    Let $(\literal \vee C)$ be a clause in the reduced formula $\formula$, where $\var(\literal)$ is a $3$-variable. Then 
    $
        \mu(\formula) - \mu(R(\formula\setminus\{\literal \vee C\})) \ge (1 + \abs{C})w_2.
    $ 
    Furthermore, if $C$ contains at least one $2$-variable, then
    $
        \mu(\formula) - \mu(R(\formula\setminus\{\literal \vee C\})) \ge (2 + \abs{C})w_2.
    $
\end{lemma}
\begin{proof}
    Let $V_2$ and $V_3$ be the sets of $2$-variables and $3$-variables in $C$, respectively. Since $\formula$ is reduced and has maximum degree $3$, we have $\abs{V_2} + \abs{V_3} = \abs{C}$. Removing $(\literal \vee C)$ reduces the degree of its constituent variables. This provides a baseline measure decrease of $\delta_3 + \abs{V_2}\delta_2 + \abs{V_3}\delta_3 = (1 + \abs{C})w_2$.
    
    Assume $C$ contains at least one $2$-variable ($\abs{V_2} \ge 1$). The removal drops the degree of every variable in $V_2$ to exactly $1$, triggering the $1$-variable reduction rule (\cref{rrule:1var}). Let $\mathcal{E}$ be the set of clauses in $\formula\setminus \{(\ell \vee C)\}$ containing at least one variable from $V_2$. Since every $2$-variable in $V_2$ has exactly one other occurrence outside $(\literal \vee C)$, $\mathcal{E}$ is non-empty. 
    
    Let $N_{ext} = \var(\mathcal{E}) \setminus V_2$ be the set of external neighbors. If $N_{ext} = \emptyset$, then $\var(\mathcal{E}) \subseteq V_2$. For any clause $E \in \mathcal{E}$, all its variables must be $2$-variables that also appear in $(\literal \vee C)$. By \cref{lem:reduced-formula}, two distinct clauses share at most one $2$-variable. Thus, $\abs{\var(E)} \le 1$, meaning $E$ is a $1$-clause, which contradicts that $\formula$ is reduced. Therefore, $\abs{N_{ext}} \ge 1$.
    
    Let $z\in N_{ext}$. The application of \cref{rrule:1var} on $V_2$ also assigns values to variables in $N_{ext}$ (including $z$). Since $z \notin V_2$, its degree after removing $(\literal \vee C)$ remains at least $2$ (if $z \in V_3 \cup \{\var(\literal)\}$, it drops from $3$ to $2$). Thus, the elimination of $z$ contributes an additional measure drop of at least $w_2$, yielding a total decrease of at least $(2 + \abs{C})w_2$.
\end{proof}

\paragraph{Step 2: Dealing with $4^+$-clauses}
We first eliminate any $4^+$-clauses containing $3$-variables. If there exists a clause $(x \vee C)$ such that $x$ is a $3$-variable and $\abs{C} \ge 3$, we apply clause branching on $(x \vee C)$: (B1) remove clause $(x\vee C)$; (B2) remove clause $(x\vee C)$ and set $x=0$ and $C=0$.

\begin{lemma}\label{lem:L-step2-result}
    The branching factor generated by Step~2 is at most $\tau(7.5, 7.5) < 1.0969$.
\end{lemma}
\begin{proof}
    Let $\formula_1 = \formula \setminus \{(x \vee C)\}$ and $\formula_2 = \formula[x=0, C=0]\setminus \{(x \vee C)\}$. Let $\Delta_i = \mu(\formula) - \mu(R(\formula_i))$ for $i \in \{1, 2\}$. Let $c_2$ and $c_3$ be the number of $2$-variables and $3$-variables in $C$, respectively. Note that $c_2+c_3=\abs{C}$.
    By \cref{lem:measure-drop-remove-clause}, we have $\Delta_1 \ge (1 + \abs{C})w_2 \ge (1 + 3)w_2 = 4w_2$. Furthermore, if $C$ contains at least one $2$-variable (i.e., $c_2\ge 1$), then $\Delta_1 \ge (2 + \abs{C})w_2 \ge (2 + 3)w_2 = 5w_2$.
    In $\formula_2$, $x$ and all variables in $C$ are assigned and eliminated. This gives $\Delta_2 \ge w_3 + c_2w_2 + c_3w_3 = (2+c_2+2c_3)w_2$. 
    
    If $c_2=0$, then $c_3=\abs{C}-c_2\ge 3$. We have $\Delta_1\ge 4w_2$ and $\Delta_2\ge (2+0+6)w_2=8w_2$. This gives $\tau(4w_2, 8w_2)=\tau(6, 12) < 1.0836$.
    
    If $c_2\ge 1$, then $\Delta_1\ge 5w_2$. Since $c_2+2c_3 = \abs{C}+c_3 \ge \abs{C}$, we obtain $\Delta_2 \ge (2+\abs{C})w_2 \ge (2+3)w_2 = 5w_2$. This gives $\tau(5w_2, 5w_2)=\tau(7.5, 7.5) < 1.0969$. 
\end{proof}

\paragraph{Step 3: Dealing with mixed variables}
We handle mixed variables, i.e., variables with both positive and negative literals. Without loss of generality, we assume that for each variable, the number of its positive literals is not less than its negative literals; otherwise, we simply flip this variable to achieve so. Let $x$ be a mixed $3$-variable and the clauses containing it are $(x\vee C_1)$, $(x\vee C_2)$, and $(\nl{x}\vee D)$. Since there are no $1$-clauses and $4^+$-clauses, $\abs{C_1}, \abs{C_2}, \abs{D} \in \{1,2\}$.

\subparagraph*{Step 3.1: Apply clause branching on $(\nl{x}\vee D)$ if $\abs{D} = 2$.}
If there is a $3$-variable $x$ such that $|D| = 2$, we apply clause branching on $(\nl{x}\vee D)$: (B1) remove clause $(\nl{x}\vee D)$; (B2) remove clause $(\nl{x}\vee D)$ and set $x=1$ and $D=0$.

\begin{lemma}\label{lem:L-step3.1-result}
    The branching factor generated by Step~3.1 is at most $\tau(6,9) < 1.0983$.
\end{lemma}
\begin{proof}
    Let $\formula_1=\formula\setminus \{(\nl{x} \vee D)\}$ and $\formula_2=\formula[x=1, D=0]\setminus\{(\nl{x} \vee D)\}$. Let $\Delta_{i}=\mu(\formula)-\mu(R(\formula_i))$ for $i\in \{1,2\}$. Let $c_2$ and $c_3$ be the number of $2$-variables and $3$-variables in $D$, respectively. Note that $c_2+c_3=\abs{D}=2$.
    
    In $\formula_1$, the clause $(\nl{x} \vee D)$ is removed. By \cref{lem:measure-drop-remove-clause}, the baseline decrease is $\Delta_1 \ge (1 + \abs{D})w_2 = 3w_2$. Furthermore, if $c_2 \ge 1$, the lemma guarantees $\Delta_1 \ge (2 + \abs{D})w_2 = 4w_2$.
    
    In $\formula_2$, $x$ and variables in $D$ are assigned. This gives a baseline measure decrease of $w_3 + c_2w_2 + c_3w_3 = (2 + c_2 + 2c_3)w_2$. Crucially, the clauses $(x \vee C_1)$ and $(x \vee C_2)$ are satisfied and removed. Since the formula is reduced, $\abs{C_1}, \abs{C_2} \ge 1$, meaning at least two literal occurrences are lost. If these literals belong to distinct variables, the measure drops by $2w_2$. If they belong to the same $3$-variable, its degree drops from $3$ to $1$, yielding $\delta_3 + \delta_2 = 2w_2$. If they belong to the same $2$-variable, this variable is dominated by $x$ (\cref{rrule:dominate}), contradicting that $\formula$ is reduced. Thus, this cascading part strictly contributes at least $2w_2$. 
    We obtain $\Delta_2 \ge (2 + c_2 + 2c_3)w_2 + 2w_2 = (4 + c_2 + 2c_3)w_2$.
    
    If $c_2 = 0$, then $c_3 = 2$. We have $\Delta_1 \ge 3w_2$ and $\Delta_2 \ge (4 + 0 + 4)w_2 = 8w_2$. This gives the branching factor $\tau(3w_2, 8w_2)=\tau(4.5, 12) < 1.0952$.
    If $c_2 \ge 1$, we have $\Delta_1 \ge 4w_2$. Since $c_2+c_3=2$, we have $c_2+2c_3 = 2+c_3 \ge 2$, yielding $\Delta_2 \ge 6w_2$. This yields the branching factor $\tau(4w_2, 6w_2)=\tau(6, 9) < 1.0983$.
\end{proof}

\subparagraph*{Step 3.2: Apply simple branching on $x$ if $\abs{D} = 1$.}
Let $D = \{y\}$. The clauses containing $x$ are $(x \vee C_1)$, $(x \vee C_2)$, and $(\nl{x} \vee y)$. 
We apply simple branching on $x$: (B1) $x=0$; (B2) $x=1$.

Before calculating the measure decrease, we establish a strict structural property. 

\begin{lemma}\label{lem:L-step3.2-extra}
    In the reduced formula $\formula$, if $D=\{y\}$, then $y \notin \var(C_1)$ or $y \notin \var(C_2)$. 
    Furthermore, for any $i \in \{1, 2\}$, if $\abs{C_i} = 1$, then $y \notin \var(C_i)$.
\end{lemma}
\begin{proof}
    Suppose for contradiction that $y \in \var(C_1)$ and $y \in \var(C_2)$. 
    If $y$ appears as a positive literal in $C_1$, the formula contains $(x \vee y \vee C_1')$ and $(\nl{x} \vee y)$. By \cref{rrule:subsumption-complementary-ltr}, literal $x$ can be resolved and removed from the first clause, contradicting that $\formula$ is reduced. Thus, $y$ must appear as $\nl{y}$ in both $C_1$ and $C_2$.
    Since its maximum degree is $3$, these three occurrences ($y$ in $D$, $\nl{y}$ in $C_1$ and $C_2$) are all its occurrences. Notice that $x$ and $\nl{y}$ always appear together with the exact same polarity. Thus, $x$ and $\nl{y}$ are twins. By \cref{rrule:twins}, $y$ can be removed, again contradicting that $\formula$ is reduced. Thus, $y \notin \var(C_1)$ or $y \notin \var(C_2)$.
    
    Now suppose $\abs{C_i} = 1$, say $C_1 = \{y_1\}$, and $y \in \var(C_1)$. This means $y_1 = y$ or $y_1 = \nl{y}$.
    {If $y_1 = y$, the formula contains $(x \vee y)$ and $(\nl{x} \vee y)$, where \cref{rrule:subsumption-complementary-ltr} can be applied, contradicting that $\formula$ is reduced.
    If $y_1 = \nl{y}$, the formula contains $(x \vee \nl{y})$ and $(\nl{x} \vee y)$, which triggers \cref{rrule:two-2cls-complementary}.} Thus, $y \notin \var(C_i)$ if $\abs{C_i}=1$.
\end{proof}

\begin{lemma}\label{lem:L-step3.2-result}
    The branching factor generated by Step~3.2 is at most $\tau(10.5, 4.5) < 1.1031$.
\end{lemma}
\begin{proof}
    Let $\formula_1 = \formula[x=1]$ and $\formula_0 = \formula[x=0]$. Let $\Delta_i = \mu(\formula) - \mu(R(\formula_i))$ for $i \in \{0, 1\}$. When $x=1$, the clauses $(x \vee C_1)$ and $(x \vee C_2)$ are satisfied, resulting in the removal of $C_1$ and $C_2$. When $x=0$, the clause $(\nl{x} \vee y)$ is satisfied and removed. Since $x$ is a $3$-variable, assigning a value to $x$ contributes exactly $w_3 = 2w_2$ to the measure decrease in both branches. Next, we analyze the measure decrease contributed by clauses $C_1, C_2$ and variable $y$. We consider two cases based on the sizes of $C_1$ and $C_2$.

    \textbf{Case 1: $\abs{C_1} = 1$ and $\abs{C_2} = 1$.} 
    Let $C_1=\{y_1\}$ and $C_2=\{y_2\}$. By \cref{lem:L-step3.2-extra}, $y_1 \neq y$ and $y_2 \neq y$. When $x=1$, the removal of $C_1$ and $C_2$ causes both $y_1$ and $y_2$ to lose occurrences, contributing at least $2w_2$ to the measure decrease. The clause $(\nl{x} \vee y)$ shrinks to $(y)$, forcing $y=1$ via \cref{rrule:1cls} and contributing at least $w_2$. Summing the contributions from $x, y_1, y_2$, and $y$, we obtain $\Delta_1 \ge w_3 + 2w_2 + w_2 = 5w_2$. Similarly, when $x=0$, {variable $y$ decreases its degree and thus contributes at least $w_2$ to $\Delta_0$}; both variables $y_1$ and $y_2$ are assigned via \cref{rrule:1cls}. Thus, $\Delta_0 \ge w_3 + w_2 + 2w_2 = 5w_2$. This yields a branching factor of $\tau(5w_2, 5w_2) = \tau(7.5, 7.5) < 1.0968$.

    \textbf{Case 2: At least one of $C_1, C_2$ has size $2$.} 
    We assume w.l.o.g.\ that $\abs{C_2} = 2$. By \cref{lem:L-step3.2-extra}, $y$ can belong to at most one of the sub-clauses $C_1, C_2$; furthermore, if $\abs{C_1}=1$, $y \notin \var(C_1)$. Thus, we can assume w.l.o.g. that $y \notin \var(C_1)$. Because $y \notin \var(C_1)$, the measure decreases jointly contributed by variable $y$ and clause $C_2$ are independent from that by variable $x$ and clause $C_1$. 
    For $i\in \{0,1\}$, let $\Delta_{i}'$ be the contribution by variable $x$ and clause $C_1$ to $\Delta_i$, and let $\Delta_{i}''$ be the contribution by variable $y$ and clause $C_2$ to $\Delta_i$. We have $\Delta_{i} \ge \Delta_{i}' + \Delta_{i}''$.
    
    We first analyze $\Delta_{1}'$ and $\Delta_{0}'$. When $x=1$, the removal of $C_1$ causes its variables to lose occurrences, contributing at least $\abs{C_1}w_2$ to $\Delta_{1}'$. When $x=0$, the clause $(x \vee C_1)$ shrinks to $C_1$. If $\abs{C_1} = 1$, \cref{rrule:1cls} can be applied, contributing at least $w_2$ to $\Delta_{0}'$. Thus, $C_1$ always contributes at least $(2-\abs{C_1})w_2$ when $x=0$. Including the contribution $w_3 = 2w_2$ from $x$, we have $\Delta_1' \ge 2w_2 + \abs{C_1}w_2$ and $\Delta_0' \ge 2w_2 + (2-\abs{C_1})w_2$.

    Next, we analyze $\Delta_{1}''$ and $\Delta_{0}''$ under two subcases. Let $z \in \var(C_2) \setminus \{y\}$.
    
    \textbf{Subcase 2.1: $y$ is a $3$-variable, or $y \notin \var(C_2)$.} When $x=1$, $y$ is forced to $1$. If $y$ is a $3$-variable or $y \notin \var(C_2)$, the combined effect of assigning $y$ and removing $C_2$ yields $\Delta_1'' \ge 3w_2$. When $x=0$, $y$ drops its degree due to the removal of $(\nl{x} \vee y)$, yielding $\Delta_0'' \ge w_2$.

    \textbf{Subcase 2.2: $y$ is a $2$-variable and $y \in \var(C_2)$.} When $x=1$, $y$ is eliminated and $z$ loses an occurrence due to the removal of $C_2$, yielding $\Delta_1'' \ge 2w_2$. When $x=0$, the $2$-variable $y$ loses an occurrence and becomes a $1$-variable. Its only remaining occurrence is $\nl{y}$ in $C_2$. By \cref{rrule:1var}, $C_2$ is falsified, and $z$ will also be assigned, yielding $\Delta_0'' \ge 2w_2$.

    Now we evaluate the total measure decrease $ \Delta_1 + \Delta_0$. The sum of the contributions from $x$ and $C_1$ is $\Delta_1' + \Delta_0' \ge (2+\abs{C_1})w_2 + (4-\abs{C_1})w_2 = 6w_2$. In both Subcase 2.1 and Subcase 2.2, the joint contributions from $y$ and $C_2$ satisfy $\Delta_1'' + \Delta_0'' \ge 4w_2$. By $\Delta_i \ge \Delta_i' + \Delta_i''$ for $i\in \{0, 1\}$, we have $\Delta_1 + \Delta_0 \ge 6w_2+4w_2= 10w_2$. 
    
    Furthermore, observing the individual bounds, $\min(\Delta_1, \Delta_0) \ge 3w_2$ holds in every case (which is tight when $\abs{C_1}=2$ and we are in Subcase 2.1). Together with $\Delta_1 + \Delta_0 \ge 10w_2$, the branching factor is at most $\tau(10w_2-3w_2, 3w_2)=\tau(7w_2, 3w_2) = \tau(10.5, 4.5) < 1.1031$.
\end{proof}

\paragraph{Step 4: Dealing with 2-clauses}
After Step 3, the formula is a positive formula where all {$3$-variables} have only positive occurrences. In the subsequent steps, we assume that a $3$-variable $x$ is contained in clauses $(x \vee C_1), (x \vee C_2)$, and $(x \vee C_3)$.

In this step, if at least one clause containing $x$ is a $2$-clause (i.e., $\min(\abs{C_1}, \abs{C_2}, \abs{C_3}) = 1$), we apply simple branching on $x$: (B1) assign $x=1$; (B2) assign $x=0$.

\begin{lemma}\label{lem:L-step4}
    The branching factor generated by Step~4 is at most $\tau(10.5, 4.5) < 1.1031$.
\end{lemma}
\begin{proof}
    Let $\formula_1 = \formula[x=1]$ and $\formula_0 = \formula[x=0]$. Let $\Delta_i = \mu(\formula) - \mu(R(\formula_i))$ for $i \in \{0, 1\}$.
    In both branches, assigning a value to the $3$-variable $x$ drops the measure by exactly $w_3 = 2w_2$. Let $c \ge 1$ be the number of clauses $C_i$ with $\abs{C_i}=1$ for $i\in \{1,2,3\}$.
    
    In branch (B1), $x=1$. The clauses $(x \vee C_1), (x \vee C_2)$, and $(x \vee C_3)$ are satisfied and removed. Consequently, the variables in $C_1, C_2, C_3$ lose a total of $\sum_{i=1}^3 \abs{C_i}$ occurrences. Each lost occurrence independently contributes at least $w_2$ to the measure drop. To verify this, consider any variable $z \in \var(C_1 \cup C_2 \cup C_3)$. If $z$ is a $3$-variable, it appears at most twice in these sub-clauses; otherwise, $z$ is dominated by $x$ and \cref{rrule:dominate} is applicable, contradicting that $\formula$ is reduced. Thus, each of its lost occurrences contributes $\delta_3 = w_2$. If $z$ is a $2$-variable, it appears at most once in these sub-clauses; otherwise, it is dominated by $x$. Thus, its single lost occurrence contributes $\delta_2 = w_2$. Therefore, $\Delta_1 \ge w_3 + (\sum_{i=1}^3 \abs{C_i}) w_2$.
    
    In branch (B2), $x=0$. The $c$ clauses of the form $(x \vee y_j)$ shrink to $1$-clauses $(y_j)$. By \cref{lem:reduced-formula}, any two $2$-clauses can share at most one common variable. Since these $c$ clauses already share the variable $x$, their remaining variables $y_j$ must be pairwise distinct. By \cref{rrule:1cls}, these $c$ distinct variables are assigned and eliminated, each contributing at least $w_2$. This yields $\Delta_0 \ge w_3 + c w_2$.
    
    Summing the bounds gives $\Delta_1+\Delta_0 \ge 2w_3 + (\sum_{i=1}^3 \abs{C_i} + c)w_2$. Since $c$ of the three sub-clauses $C_1,C_2,C_3$ have length $1$ and the remaining $3-c$ have length $2$, we have $\sum_{i=1}^3 \abs{C_i} = c + 2(3-c) = 6 - c$. This implies $\sum_{i=1}^3 \abs{C_i} + c = 6$. With $w_3 = 2w_2$, we get $\Delta_1 + \Delta_0 \ge 10w_2$. Moreover, the equality $\sum_{i=1}^3 \abs{C_i} = 6 - c \ge 3$ (since $c \le 3$) and $c \ge 1$ guarantee that $\Delta_1 \ge 5w_2$ and $\Delta_0 \ge 3w_2$. Since $\Delta_1+\Delta_0 \ge 10w_2$ and $\min(\Delta_1, \Delta_0) \ge 3w_2$, the branching factor generated by this step is at most $\tau(10w_2-3w_2, 3w_2) = \tau(7w_2, 3w_2) = \tau(10.5, 4.5) < 1.1031$.
\end{proof}

\paragraph{Step 5: Dealing with 3-clauses}
After Step 4, all $3$-variables appear positively, and any clause containing a $3$-variable has a length of exactly $3$. 
Before presenting the specific branching steps, we first analyze the measure decrease after assigning $x=1$. Let $x$ be a $3$-variable appearing in $(x \vee C_1), (x \vee C_2)$, and $(x \vee C_3)$.
We define $Y_x \subseteq \var(C_1 \cup C_2 \cup C_3)$ as the set of variables that have exactly one occurrence outside these three clauses. That is, $Y_x$ consists of $2$-variables that appear exactly once in $C_1 \cup C_2 \cup C_3$, and $3$-variables that appear exactly twice in $C_1 \cup C_2 \cup C_3$. 
Let $\mathcal{R}_x$ be the set of clauses in $\formula$ outside of $\{(x \vee C_i)\}_{i=1}^3$ that contain the remaining occurrences of variables in $Y_x$, i.e., $\mathcal{R}_x = \{D \in \formula \setminus \{(x \vee C_i)\}_{i=1}^3 \mid \var(D) \cap Y_x \neq \emptyset\}$. 
Furthermore, let $\mathrm{Ext}_x = \var(\mathcal{R}_x) \setminus Y_x$ be the set of their external neighbors. 

Observe that when assigning $x=1$, the variables in $Y_x$ lose all but their one remaining occurrence in $\mathcal{R}_x$, thus becoming $1$-variables in $\formula[x=1]$. By \cref{rrule:1var}, all variables in $\mathrm{Ext}_x$ will be assigned values and give us extra gain in the measure decrease. Formally, we have:

\begin{lemma}\label{lem:unified-x1}
    Let $x$ be a $3$-variable appearing in $(x \vee C_1), (x \vee C_2)$, and $(x \vee C_3)$. It holds that $\mu(\formula) - \mu(R(\formula[x=1])) \ge 8w_2 + \abs{\mathrm{Ext}_x}w_2$.
\end{lemma}
\begin{proof}
    Assigning $x=1$ satisfies and removes the three clauses $(x \vee C_i)$. Note that any $2$-variable can appear at most once in $C_1, C_2, C_3$, and any $3$-variable can appear at most twice in $C_1, C_2, C_3$. Otherwise, \cref{rrule:dominate} is applicable, which contradicts that $\formula$ is reduced. Thus, each literal occurrence in $C_1,C_2,C_3$ independently contributes at least $\min\{\delta_3,\delta_2,w_2\}=w_2$ to the measure decrease. Consequently, the removal of $(x \vee C_1), (x \vee C_2)$, $(x \vee C_3)$ contributes $w_3 + \sum_{i=1}^{3}\abs{C_i}=8w_2$ since $w_3=2w_2$ and $\abs{C_1}=\abs{C_2}=\abs{C_3}=2$.
    
    After removing clauses $(x \vee C_1), (x \vee C_2)$, $(x \vee C_3)$, variables in $Y_x$ become $1$-variables, which triggers \cref{rrule:1var} to assign values to variables in $\mathrm{Ext}_x$. This contributes an additional $\abs{\mathrm{Ext}_x}w_2$ measure drop. Summing up both parts, we have $\mu(\formula) - \mu(R(\formula[x=1]))\ge 8w_2+\abs{\mathrm{Ext}_x}w_2$, which completes the proof.
\end{proof}

\subparagraph*{Step 5.1: Apply simple branching on a variable $x$ with $\abs{\mathrm{Ext}_{x}} \ge 1$.}
If there is a $3$-variable $x$ satisfying $\abs{\mathrm{Ext}_{x}} \ge 1$, we apply simple branching on $x$: (B1) assign $x=1$; (B2) assign $x=0$. Guided by \cref{lem:unified-x1}, we have the following result.

\begin{lemma}\label{lem:L-step5.1}
    The branching factor generated by Step~5.1 is at most $\tau(13.5, 3) < 1.1052$.
\end{lemma}
\begin{proof}
    In (B1), by \cref{lem:unified-x1}, since $\abs{\mathrm{Ext}_{x}} \ge 1$, we have $\Delta_1 \ge 8w_2 + \abs{\mathrm{Ext}_{x}}w_2 \ge 9w_2$. In branch (B2), removing $x=0$ yields a measure drop of $\Delta_0 \ge w_3 = 2w_2$. This yields the branching factor $\tau(9w_2, 2w_2)= \tau(13.5, 3) < 1.1052$.
\end{proof}

We then distinguish $3$-variables by \emph{proper} and \emph{non-proper} based on the composition of the clauses containing it:
\begin{definition}[proper variables]
    Let $x$ be a $3$-variable appearing in $(x \vee C_1), (x \vee C_2)$, $(x \vee C_3)$. $x$ is called \emph{proper} if the following two conditions hold:
    \begin{itemize}
        \item $C_1,C_2,C_3$ are mutually disjoint (i.e., $\var(C_i) \cap \var(C_j) = \emptyset$ for $i \neq j$); 
        \item all variables in $C_1 \cup C_2 \cup C_3$ are $3$-variables. 
    \end{itemize}
    Otherwise, $x$ is \emph{non-proper}.
\end{definition}

The following lemma shows that Step~5.1 in fact handles all non-proper $3$-variables. 
\begin{lemma}\label{lem:structural-dichotomy}
    Let $\formula$ be a reduced formula. If there exists at least one non-proper $3$-variable in $\formula$, then there exists a $3$-variable $x^*$ such that $\abs{\mathrm{Ext}_{x^*}} \ge 1$.
\end{lemma}
\begin{proof}
    Let $x$ be a non-proper $3$-variable in $\formula$, and let $(x\lor C_1), (x\lor C_2), (x\lor C_3)$ be the three clauses containing $x$. Let $Y_{x,j} \subseteq \var(C_1 \cup C_2 \cup C_3)$ for $j \in \{2, 3\}$ be the variables appearing exactly $j-1$ times in the sub-clauses. Because $x$ is non-proper, its local structure must fall into one of the following three exhaustive cases. We outline the specific conclusion we prove for each case:
    \begin{itemize}
        \item \textbf{Case 1:} $\abs{Y_{x,3}}\ge 1$. Equivalently, there is an overlap among the sub-clauses of $x$ (i.e., $\var(C_i) \cap \var(C_j) \neq \emptyset$ for some $i \neq j$). The equivalence follows from the fact that a $2$-variable can only appear once in $C_1,C_2,C_3$ (otherwise, it is dominated by $x$, triggering \cref{rrule:dominate} and contradicting that $\formula$ is reduced). In this case, we prove that $\abs{\mathrm{Ext}_x} \ge 1$.
        \item \textbf{Case 2:} Case 1 does not hold (i.e., $\abs{Y_{x,3}}=0$), and there exists a clause $(x\lor y\lor z)$ such that $y$ is a $3$-variable and $z$ is a $2$-variable. We show that either $\abs{\mathrm{Ext}_x} \ge 1$ or $\abs{\mathrm{Ext}_y} \ge 1$.
        \item \textbf{Case 3:} Cases 1 and 2 do not hold, and $\abs{Y_{x,2}}\ge 1$, i.e., there exists at least one $2$-variable in $C_1\cup C_2\cup C_3$. We prove that $\abs{\mathrm{Ext}_x} \ge 1$.
    \end{itemize}

    \textbf{Case 1:} $\abs{Y_{x,3}}\ge 1$. Suppose for contradiction that $\abs{\mathrm{Ext}_x} = 0$, which implies all clauses in $\mathcal{R}_x$ are formed exclusively by variables from $Y_x$, i.e., $\var(\mathcal{R}_x) = Y_x$. 
    
    Consider the subformula $\formula' = \{(x \vee C_i)\}_{i=1}^3 \cup \mathcal{R}_x$. Note that all occurrences of variables in $Y_x \cup \{x\}$ are strictly contained within $\var(\formula')$. Therefore, the intersection between $\formula'$ and the rest of the formula, $\var(\formula') \cap \var(\formula \setminus \formula')$, can only consist of variables in $\var(C_1 \cup C_2 \cup C_3) \setminus Y_x$. Since variables in $Y_{x,2}$ have $1$ occurrence and variables in $Y_{x,3}$ have $2$ occurrences in the sub-clauses, the number of remaining literal occurrences in the sub-clauses is exactly $6 - \abs{Y_{x,2}} - 2\abs{Y_{x,3}}$. This provides an upper bound on the intersection size:
    \begin{align*}
        \abs{\var(\formula') \cap \var(\formula \setminus \formula')} \le 6 - \abs{Y_{x,2}} - 2\abs{Y_{x,3}}.
    \end{align*}
    To avoid violating Property 5 of \cref{lem:reduced-formula}, the intersection size should be at least $2$. 
    If $\abs{Y_{x,3}} = 3$, the intersection size is at most $6 - 0 - 6 = 0$. Thus, $\abs{Y_{x,3}}\le 2$.
    
    If $\abs{Y_{x,3}} = 2$, we must have $\abs{Y_{x,2}} = 0$ (otherwise, the intersection size is at most $6 - 1 - 4 = 1$). Let $Y_{x,3} = \{y_1, y_2\}$. These two variables occupy $4$ literal occurrences in $C_1,C_2,C_3$. Since $\abs{C_1}=\abs{C_2}=\abs{C_3}=2$ and each variable can appear at most once in a clause, by the Pigeonhole Principle, there exists some $C_i$ such that $C_i = \{y_1, y_2\}$. Meanwhile, since $\abs{\mathrm{Ext}_x} = 0$ and $\abs{Y_{x,2}} = 0$, $\mathcal{R}_x$ is formed exclusively by $y_1$ and $y_2$. Thus, any clause $D \in \mathcal{R}_x$ must be a subset of $\{y_1, y_2\}$, implying $D \subseteq C_i$. However, this would trigger \cref{rrule:subsumption} and contradict that $\formula$ is reduced.
    
    If $\abs{Y_{x,3}} = 1$, we must have $\abs{Y_{x,2}} \le 2$. Otherwise, if $\abs{Y_{x,2}} \ge 3$, the intersection size is at most $6 - 3 - 2 = 1$. Let $Y_{x,3} = \{y\}$ and assume without loss of generality that $y \in C_1$ and $y \in C_2$. Let $D\in \mathcal{R}_x$ be the clause containing $y$. Since $y$ is a $3$-variable, we have $\abs{D}=3$. Thus, $D$ requires two additional distinct variables from $Y_x$ since {$\var(\mathcal{R}_x) \subseteq Y_x$}. As $y$ is the only variable in $Y_{x,3}$, these two additional variables must belong to $Y_{x,2}$, and so $\abs{Y_{x,2}} \ge 2$. Therefore, $\abs{Y_{x,2}} = 2$ and $\mathcal{R}_x=\{D\}$.

    {
    Let $Y_{x,2} = \{u, v\}$, then $D = (y \lor \literal_u \lor \literal_v)$, where $\literal_u$ and $\literal_v$ are literals of variables $u$ and $v$, respectively. 
    Note that variables $u, v \in \var(C_1 \cup C_2 \cup C_3)$. 
    If variable $u$ appears in $C_1$, let $\literal_u'$ be its literal in $C_1$, then $(x \vee C_1) = (x \vee y \vee \literal_u')$. 
    Since $u$ is a $2$-variable and the positive literal $y$ appears in every clause containing variable $u$ (namely $(x \vee C_1)$ and $D$), the variable $u$ is dominated by literal $y$.
    This triggers \cref{rrule:dominate} and contradicts that $\formula$ is reduced. 
    Thus, variable $u \notin \var(C_1)$, and by symmetry, $u \notin \var(C_2)$. 
    The exact same logic applies to $v$, meaning variable $v \notin \var(C_1 \cup C_2)$.
    Therefore, variables $u, v \in \var(C_3)$. Let $\literal_u'$ and $\literal_v'$ be their respective literals in $C_3$, meaning $(x \vee C_3) = (x \vee \literal_u' \vee \literal_v')$. 
    However, this contradicts \cref{lem:reduced-formula}, which says that every two clauses can share at most one $2$-variable in a reduced formula. 
    Thus, we must have $\abs{\mathrm{Ext}_x} \ge 1$.
    }
    
    \textbf{Case 2:} $\abs{Y_{x,3}}=0$, and there exists a clause $(x\lor y\lor z)$ such that $y$ is a $3$-variable and $z$ is a $2$-variable.
    Suppose for contradiction that $\abs{\mathrm{Ext}_x} = 0$ and $\abs{\mathrm{Ext}_y} = 0$. 
    Since $z$ is a $2$-variable, $z \in Y_{x,2}$, meaning its only other occurrence must be in some clause $E \in \mathcal{R}_x$. Since $\formula$ is reduced, $\abs{E} \ge 2$, so there exists another variable $z' \in E \setminus \{z\}$. 
    Because $\abs{\mathrm{Ext}_x} = 0$ and $\abs{Y_{x,3}} = 0$, we have $\var(\mathcal{R}_x) = Y_{x,2}$. Thus, $z'$ must be a $2$-variable in $Y_{x,2}$, implying that its only other occurrence should co-occur with $x$. 
    On the other hand, note that $E$ also belongs to $\mathcal{R}_y$. Thus, $z'\in Y_{y,2}$, implying that the only other occurrence of $z'$ should co-occur with $y$. 
    Since $z'$ is a $2$-variable and one of its occurrences is in $E$, the other one remaining occurrence should co-occur with both $x$ and $y$, forming a clause $(x \vee y \vee z')$. 
    Furthermore, because $z' \neq z$, this clause must be distinct from $(x \vee y \vee z)$. However, this implies $y \in Y_{x,3}$, which contradicts $\abs{Y_{x,3}} = 0$. 
    Thus, we have either $\abs{\mathrm{Ext}_x} \ge 1$ or $\abs{\mathrm{Ext}_y} \ge 1$.
    
    \textbf{Case 3:} {Cases 1 and 2 do not hold, and $\abs{Y_{x,2}} \ge 1$.}
    Suppose for contradiction that $\abs{\mathrm{Ext}_x} = 0$, i.e., $\var(\mathcal{R}_x)\subseteq Y_{x}$. Since Case 1 does not hold (i.e., $\abs{Y_{x,3}}=0$), $Y_{x}=Y_{x,2}$. Recall that $x$ is contained in clauses $(x\lor C_1), (x\lor C_2), (x\lor C_3)$. Let $\mathcal{T}\subseteq \{(x\lor C_i)\}_{i=1}^{3}$ be the set of clauses such that $C_i$ contains $2$-variables. Because Case 2 does not hold, for every $C_i$, it holds that either $\var(C_i)\subseteq Y_{x,2}$ or $\var(C_i)\cap Y_{x,2}=\emptyset$. Thus, $\var(\mathcal{T})=Y_{x,2}\cup \{x\}=Y_{x}\cup \{x\}$. Consider the subformula $\formula'' = \mathcal{T}\cup \mathcal{R}_{x}$. Since $\var(\mathcal{R}_x)\subseteq Y_{x}$, the set of variables in this subformula is exactly $\var(\formula'') = Y_x \cup \{x\}$. 
    Because all variables in $Y_x$ are $2$-variables whose two occurrences are contained within $\mathcal{T}$ and $\mathcal{R}_x$, they do not appear anywhere in $\formula \setminus \formula''$. Thus, $\var(\formula'') \cap \var(\formula \setminus \formula'') \subseteq \{x\}$, which violates Property 5 of \cref{lem:reduced-formula} and contradicts that $\formula$ is reduced. Thus, we must have $\abs{\mathrm{Ext}_x} \ge 1$.
    
    In all exhaustive cases, we identify a $3$-variable $x^*$ satisfying $\abs{\mathrm{Ext}_{x^*}} \ge 1$.
\end{proof}

\subparagraph*{Step 5.2: Apply variable branching on proper variables.}
By \cref{lem:structural-dichotomy}, if the algorithm reaches this step, all $3$-variables in the formula are proper. We pick up an arbitrary $3$-variable $x$ (which is contained in $(x \vee C_1), (x \vee C_2)$, and $(x \vee C_3)$), and apply the branching rule in \cref{lem:variable-branching} on $x$: (B1) assign $x=1, C_1=0$; (B2) assign $x=1, C_2=0$, and add clause $C_1$; (B3) assign $x=1, C_3=0$, and add clauses $C_1$ and $C_2$.

\begin{lemma}\label{lem:L-step5.2}
    The branching factor generated by Step~5.2 is at most $\tau(15, 12, 9) < 1.0983$.
\end{lemma}
\begin{proof}
    In all branches, $x$ is assigned to $1$, contributing a measure drop of $w_3$. 
    Since $x$ is proper, $C_1$, $C_2$, and $C_3$ are mutually disjoint and contain only $3$-variables. 
    
    In branch (B1), $C_1$ is falsified. Since $\abs{C_1} = 2$ and it contains only $3$-variables, assigning the variables in $C_1$ yields a measure drop of $2w_3$. The clauses $(x \vee C_2)$ and $(x \vee C_3)$ are satisfied and removed. Because $C_1, C_2, C_3$ are mutually disjoint, this removal causes the four variables in $C_2$ and $C_3$ to lose occurrences, contributing at least $4\delta_3$. Thus, $\Delta_1 \ge w_3 + 2w_3 + 4\delta_3 = 3w_3 + 4w_2 = 10w_2$.
    
    In branch (B2), $C_2$ is falsified. Assigning its variables gives $2w_3$. The clause $(x \vee C_3)$ is satisfied and removed, causing the variables in $C_3$ to lose occurrences and decreasing the measure by at least $2\delta_3$. Adding clause $C_1$ preserves the occurrences of its variables. Thus, $\Delta_2 \ge w_3 + 2w_3 + 2\delta_3 = 3w_3 + 2w_2 = 8w_2$.
    
    In branch (B3), $C_3$ is falsified. Assigning the variables in $C_3$ decreases the measure by at least $2w_3$. Adding clauses $C_1$ and $C_2$ preserves the degrees of their variables. Thus, $\Delta_3 \ge w_3 + 2w_3 = 6w_2$.
    
    Therefore, the branching factor is at most $\tau(10w_2, 8w_2, 6w_2) = \tau(15, 12, 9) < 1.0983$.
\end{proof}

\subsubsection{Step 6: Reduce to \ParityTwoOccSAT{}}
Now all variables are $2$-variables.
We have $\mu(\formula)=w_2n(\formula)$.
Applying the $\bigO{1.1193^{n(\formula)}}$-time algorithm for \ParityTwoOccSAT{} (\cref{thm:result-2occ-n}) to solve the remaining part takes $\bigO{(1.1193^{\frac{1}{w_2}})^{\mu(\formula)}}\subseteq \bigO{1.0781^{\mu(\formula)}}
$ time.

\subsubsection{The Final Result}
As summarized in Table~\ref{tab:result-L-overview}, the largest branching factor among all branching steps 1--5 is $1.1052$. All branching steps take polynomial space. Step 6 takes $\bigO{1.0781^{\mu(\formula)}}$ time and polynomial space. Thus, we obtain 
\begin{theorem}
    \ParitySAT{} can be solved in $\bigO{1.1052^L}$ time and polynomial space, where $L$ is the formula length.
\end{theorem}

%% file: tables/result-L-overview.tex
\begin{tabular}{llll}
    \toprule
    \textbf{Steps} & \textbf{Branching Objects} & \textbf{Vectors} & \textbf{Factors} \\
    \midrule
    Step 1 & $4^+$-variables & $(12, 4)$ & $1.1003$ \\
    \midrule
    Step 2 & $4^+$-clauses containing $3$-variables & $(7.5, 7.5)$ & $1.0969$ \\
    \midrule
    \multirow{2}{*}{Step 3} & \multirow{2}{*}{\makecell[l]{$3$-variables with both positive\\ literals and negative literals}} & Step 3.1: $(6, 9)$ & $1.0983$ \\
           &                    & Step 3.2: $(10.5, 4.5)$ & $1.1031$ \\
    \midrule
    Step 4 & $3$-variables contained in $2$-clauses  & $(10.5, 4.5)$ & $1.1031$ \\
    \midrule
    \multirow{2}{*}{Step 5} & \multirow{2}{*}{Remaining $3$-variables} & Step 5.1: $(13.5, 3)$ & $1.1052$* \\
           &                                         & Step 5.2: $(15, 12, 9)$ & $1.0983$ \\
    \midrule
    Step 6 & \makecell[l]{$2$-variables\\ (reduce to \ParityTwoOccSAT{})} & \makecell[l]{$1.1193^{\frac{1}{1.5}}$\\ (\cref{thm:result-2occ-n})} & $1.0781$ \\
    \bottomrule
\end{tabular}

%% file: sections/conclusion.tex
\section{Conclusion and Discussion}\label{sec:conclusion}
In this paper, we investigated algorithms for \ParitySAT{} and its bounded-occurrence version, \ParitydOccSAT{}. By systematically exploiting the properties of parity, we developed polynomial-space algorithms that achieve better running time bounds than the best-known bounds for their exact counting counterparts. Specifically, we broke the $2^m$-barrier for \ParitydOccSAT{} for any fixed $d$ with an $\bigOstar{2^{m(1-1/\bigO{d})}}$-time algorithm, achieved refined bounds of $\bigO{\ResultTwoOccN}$ and $\bigOstar{\ResultTwoOccM}$ time for \ParityTwoOccSAT{}, and established an $\bigO{1.1052^L}$-time algorithm for general \ParitySAT{}.

A natural question is whether one can break the $2^m$-barrier for \#$d$-occ-SAT, even for the base case of $d=2$. In the parity setting, \cref{lem:reduction-positive} allows us to safely restrict our attention to positive formulas. It is worth noting that since the clause branching rule in \cref{lem:clause-branching} inherently preserves the exact number of satisfying assignments, a counting variant of \cref{lem:reduction-positive} indeed holds. However, our subsequent algorithms for handling these positive formulas rely fundamentally on parity: {the simple $\bigOstar{2^{m(1-1/\bigO{2^d})}}$-time algorithm} applies the parity-specific variable branching in \cref{lem:variable-branching}, while {the improved $\bigOstar{2^{m(1-1/\bigO{d})}}$-time algorithm} exploits the parity equivalence between hitting sets and set covers. Neither of these techniques can be directly lifted to the exact counting counterpart.

%% file: sections/appendix.tex
\appendix
\newpage
\input{sections/rETH-hard}

\section{Monien-Preis bound for cubic multigraphs}\label{appendix:bisection-multigraphs}
\begin{theorem}[{see also \cite{journals/corr/abs-2201-03220}}]
    For any $\varepsilon > 0$, there exists an integer $n_{\varepsilon}$ such that for any cubic multigraph $G$ with $|V(G)| \geq n_{\varepsilon}$, the bisection width of $G$ is at most $(1/6 + \varepsilon)|V(G)|$.
\end{theorem}

\begin{proof}
    Let $G$ be a cubic multigraph. We obtain a simple sub-cubic graph $G'$ by replacing each set of parallel edges in $G$ with a single edge. As noted in \cite{journals/ipl/FominH06,journals/talg/GaspersS17}, the Monien-Preis bound holds for simple sub-cubic graphs. Let $\varepsilon' = \varepsilon/2$. For $n \ge n_{\varepsilon'}$, $G'$ admits a bisection $(A', B')$ such that $|E_{G'}(A', B')| \leq (1/6 + \varepsilon')n$, where $E_{G'}(A', B')$ is the set of edges in $G'$ with one endpoint in $A'$ and the other in $B'$. Let $n_{\varepsilon} = \max(n_{\varepsilon'}, 6/\varepsilon)$. We show that for $n \geq n_{\varepsilon}$, the cut size in $G$ is at most $(1/6 + \varepsilon)n$.

    We first adjust the partition on the simple graph $G'$. Suppose there exists an edge $(u, v)$ in the cut $E_{G'}(A', B')$ that corresponds to parallel edges in the original multigraph $G$ (i.e., the multiplicity $m_G(u,v) \geq 2$). We resolve this by moving the endpoint residing in the currently larger partition (say $u$) to the smaller partition.

    Since $G$ is cubic and $m_G(u,v) \ge 2$, the vertex $u$ has at most one neighbor in $G$ other than $v$. Consequently, in the simple graph $G'$, $u$ is incident to $v$ and at most one other vertex $w$. Moving $u$ to the other partition removes the edge $(u,v)$ from the cut and introduces at most one new edge $(u,w)$ into the cut. Thus, this move does not increase the cut size $|E_{G'}|$. By always moving the vertex from the larger partition to the smaller one, we ensure that the imbalance $||A|-|B||$ remains at most 2 throughout the process of eliminating cut edges with $m_G \ge 2$.

    Let the resulting partition be $(A, B)$. At this stage, every edge crossing the cut in $G'$ has multiplicity $1$ in the original multigraph $G$. Therefore, the cut size in $G$ is exactly equal to the cut size in $G'$, yielding $|E_G(A, B)| = |E_{G'}(A, B)| \leq |E_{G'}(A', B')| \leq (1/6 + \varepsilon')n$.

    If the partition $(A, B)$ is not already a bisection, the imbalance is at most $2$. We move at most one additional vertex to restore perfect balance (i.e., $||A|-|B|| \le 1$). Since the maximum degree in $G$ is $3$, this final adjustment increases the cut size by at most $3$. Given $n \geq n_{\varepsilon} \geq 6/\varepsilon$, we have $3 \leq (\varepsilon/2)n$. Consequently, the bisection width of $G$ is bounded by:
    \[
    (1/6 + \varepsilon/2)n + 3 \leq (1/6 + \varepsilon/2)n + (\varepsilon/2)n = (1/6 + \varepsilon)n.
    \]
    This completes the proof.

\end{proof}

%% file: sections/rETH-hard.tex
\section{A Hardness Result for Parity-2-occ-SAT}\label{sec:rETH-hard}
The randomized Exponential Time Hypothesis (rETH) states that there is a constant $\varepsilon>0$ such that no randomized algorithm can decide $3$-SAT in time $\bigOstar{2^{\varepsilon n}}$ with error probability at most $1/3$ {(see also \cite[Section 1.3]{journals/talg/DellHMTW14})}.
\begin{theorem}
    Under rETH, there is no randomized algorithm that solves \ParityTwoOccSAT{} {on positive formulas} in $2^{o(n)}$ time, where $n$ is the number of variables.
\end{theorem}
\begin{proof}
   Let $G=(V,E)$ be a graph with $n_v = |V|$ vertices and $m = |E|$ edges. 
    A vertex cover of $G$ is a set of vertices that includes at least one endpoint of every edge of $G$. An independent set of $G$ is a set of vertices such that there is no edge between any two vertices. It is a basic property that $S \subseteq V$ is a vertex cover if and only if $V \setminus S$ is an independent set.
    An edge cover of $G$ is a set of edges such that every vertex of $G$ is an endpoint of at least one edge of the set. For a vertex $v \in V$, we denote by $\delta(v) = \{e \in E \mid v \in e\}$ the set of edges incident to $v$. Parity-Vertex-Cover (resp., Parity-Independent-Set, Parity-Edge-Cover) is the problem of determining whether the number of vertex covers (resp., independent sets, edge covers) of $G$ is odd.
    Let $\#vc(G)$ denote the number of vertex covers of $G$ and $\#ec(G)$ the number of edge covers of $G$.

    It is known that under rETH, Parity-Independent-Set cannot be solved in $2^{o(m)}$ time, where $m$ is the number of edges~{\cite[Proof of Theorem 1.2]{journals/talg/DellHMTW14}}. 
    Since the complement of an independent set is a vertex cover, Parity-Independent-Set is equivalent to Parity-Vertex-Cover, which also implies a $2^{o(m)}$ lower bound for the latter. 
    We first show that Parity-Edge-Cover is equivalent to Parity-Vertex-Cover. Then, we reduce Parity-Edge-Cover to Parity-2-occ-SAT {on positive formulas} in a natural way.
    We claim that
    \begin{equation}
        \#vc(G) \equiv \#ec(G) \pmod 2. \label{eq:vc=ec}
    \end{equation}
    Then, the equivalence of Parity-Vertex-Cover and Parity-Edge-Cover follows from Eq.~\ref{eq:vc=ec}.
    The correctness of Eq.~\ref{eq:vc=ec} follows from~\cite[Theorem 4.3]{journals/talg/CyganDLMNOPSW16}. Here, we note the proof for our case for completeness. Consider the following Inclusion-Exclusion formulation of $\#ec(G)$:
    \begin{align*}
        \#ec(G)
        &= 2^{|E(G)|} - \sum_{u\in V(G)}2^{|E(G-\{u\})|} + \sum_{u,v\in V(G)}2^{|E(G-\{u,v\})|} - \cdots \\
        % \sum_{u,v,w\in V(G)}2^{|E(G-\{u,v,w\})|} + 
        &=\sum_{S \subseteq V(G)} (-1)^{|S|}\cdot 2^{|E(G-S)|}.
    \end{align*}
    Thus, we have
    \[
        \#ec(G)\equiv \sum_{\substack{S \subseteq V(G)\\|E(G-S)|=0}} 1 \equiv \#vc(G) \pmod 2,
    \]
    which completes the proof of Eq.~\ref{eq:vc=ec}.
    Next, we reduce Parity-Edge-Cover to Parity-2-occ-SAT.

    Given a graph $G=(V,E)$, for each edge $e\in E$ we create a variable $x_e$. The formula is then defined as $\phi_{G}=\bigwedge_{v\in V(G)}(\bigvee_{e\in \delta(v)} x_e)$.
    Clearly, each variable appears twice {positively}, and each satisfying assignment to $\phi_{G}$ corresponds to an edge cover of $G$ and vice versa. 

    In terms of parameters, the number of variables in $\phi_G$ is $n = |E| = m$. Since Parity-Vertex-Cover on $m$ edges requires $2^{\Omega(m)}$ time under rETH and our reduction is parity-preserving with $n = m$, any algorithm solving Parity-2-occ-SAT in $2^{o(n)}$ time would directly imply a $2^{o(m)}$ algorithm for Parity-Vertex-Cover, contradicting rETH.
\end{proof}
{We note that the above reduction also applies to show the $\oplus${\sf P}-Completeness of \ParityTwoOccSAT{} on positive formulas.}

%% file: _ref.bib
@inproceedings{conf/stoc/Cook71,
  author    = {Stephen A. Cook},
  title     = {The Complexity of Theorem-Proving Procedures},
  booktitle = {Proceedings of the 3rd Annual {ACM} Symposium on Theory of Computing ({STOC} 1971)},
  pages     = {151--158},
  year      = {1971},
  doi       = {10.1145/800157.805047}
}

@article{journals/tcs/Valiant79,
  author    = {Leslie G. Valiant},
  title     = {The Complexity of Computing the Permanent},
  journal   = {Theor. Comput. Sci.},
  volume    = {8},
  pages     = {189--201},
  year      = {1979},
  doi       = {10.1016/0304-3975(79)90044-6}
}

@inproceedings{conf/focs/BacchusDP03,
  author    = {Fahiem Bacchus and Shannon Dalmao and Toniann Pitassi},
  title     = {Algorithms and Complexity Results for {\#}SAT and Bayesian Inference},
  booktitle = {Proceedings of the 44th Annual Symposium on Foundations of Computer Science ({FOCS} 2003)},
  pages     = {340--351},
  year      = {2003},
  doi       = {10.1109/SFCS.2003.1238208}
}

@article{journals/ai/Roth96,
  author    = {Dan Roth},
  title     = {On the Hardness of Approximate Reasoning},
  journal   = {Artif. Intell.},
  volume    = {82},
  number    = {1-2},
  pages     = {273--302},
  year      = {1996},
  doi       = {10.1016/0004-3702(94)00092-1}
}

@inproceedings{conf/aaai/SangBK05,
  author    = {Tian Sang and Paul Beame and Henry A. Kautz},
  title     = {Performing Bayesian Inference by Weighted Model Counting},
  booktitle = {Proceedings of the 20th National Conference on Artificial Intelligence ({AAAI} 2005)},
  pages     = {475--482},
  year      = {2005}
}

@article{journals/ai/ChaviraD08,
  author    = {Mark Chavira and Adnan Darwiche},
  title     = {On probabilistic inference by weighted model counting},
  journal   = {Artif. Intell.},
  volume    = {172},
  number    = {6-7},
  pages     = {772--799},
  year      = {2008},
  doi       = {10.1016/J.ARTINT.2007.11.002}
}

@inproceedings{conf/aaai/Duenas-OsorioMP17,
  author    = {Leonardo Due{\~{n}}as{-}Osorio and Kuldeep S. Meel and Roger Paredes and Moshe Y. Vardi},
  title     = {Counting-Based Reliability Estimation for Power-Transmission Grids},
  booktitle = {Proceedings of the 31st {AAAI} Conference on Artificial Intelligence ({AAAI} 2017)},
  pages     = {4488--4494},
  year      = {2017},
  doi       = {10.1609/AAAI.V31I1.11178}
}

@inproceedings{conf/sat/NarodytskaSMIM19,
  author    = {Nina Narodytska and Aditya A. Shrotri and Kuldeep S. Meel and Alexey Ignatiev and Jo{\~{a}}o Marques{-}Silva},
  title     = {Assessing Heuristic Machine Learning Explanations with Model Counting},
  booktitle = {Proceedings of the 22nd International Conference on Theory and Applications of Satisfiability Testing ({SAT} 2019)},
  pages     = {267--278},
  year      = {2019},
  doi       = {10.1007/978-3-030-24258-9\_19}
}

@article{journals/tcs/ValiantV86,
  author    = {Leslie G. Valiant and Vijay V. Vazirani},
  title     = {{NP} is as easy as detecting unique solutions},
  journal   = {Theor. Comput. Sci.},
  volume    = {47},
  number    = {3},
  pages     = {85--93},
  year      = {1986},
  doi       = {10.1016/0304-3975(86)90135-0}
}

@inproceedings{conf/tcs/PapadimitriouZ83,
  author    = {Christos H. Papadimitriou and Stathis Zachos},
  title     = {Two remarks on the power of counting},
  booktitle = {Proceedings of the 6th {GI} Conference on Theoretical Computer Science},
  pages     = {269--276},
  year      = {1983},
  doi       = {10.1007/BFb0009651}
}

@article{journals/siamcomp/Toda91,
  author    = {Seinosuke Toda},
  title     = {{PP} is as Hard as the Polynomial-Time Hierarchy},
  journal   = {{SIAM} Journal on Computing},
  volume    = {20},
  number    = {5},
  pages     = {865--877},
  year      = {1991},
  doi       = {10.1137/0220053}
}

@article{journals/jcss/ImpagliazzoP01,
  author    = {Russell Impagliazzo and Ramamohan Paturi},
  title     = {On the Complexity of k-SAT},
  journal   = {J. Comput. Syst. Sci.},
  volume    = {62},
  number    = {2},
  pages     = {367--375},
  year      = {2001},
  doi       = {10.1006/JCSS.2000.1727}
}

@inproceedings{conf/soda/ImpagliazzoMP12,
  author    = {Russell Impagliazzo and William Matthews and Ramamohan Paturi},
  title     = {A satisfiability algorithm for AC\({}^{\mbox{0}}\)},
  booktitle = {Proceedings of the 23rd Annual {ACM-SIAM} Symposium on Discrete Algorithms ({SODA} 2012)},
  pages     = {961--972},
  year      = {1992},
  doi       = {10.1137/1.9781611973099.77}
}

@inproceedings{conf/soda/ChanW16,
  author    = {Timothy M. Chan and Ryan Williams},
  title     = {Deterministic APSP, Orthogonal Vectors, and More: Quickly Derandomizing Razborov-Smolensky},
  booktitle = {Proceedings of the 27th Annual {ACM-SIAM} Symposium on Discrete Algorithms ({SODA} 2016)},
  pages     = {1246--1255},
  year      = {2016},
  doi       = {10.1137/1.9781611974331.CH87}
}

@inproceedings{conf/iwpec/Wahlstrom08,
  author    = {Magnus Wahlstr{\"{o}}m},
  title     = {A Tighter Bound for Counting Max-Weight Solutions to 2SAT Instances},
  booktitle = {Proceedings of the 3rd International Workshop on Parameterized and Exact Computation ({IWPEC} 2008)},
  pages     = {202--213},
  year      = {2008},
  doi       = {10.1007/978-3-540-79723-4\_19}
}

@book{series/txtcs/FominK10,
  author    = {Fedor V. Fomin and Dieter Kratsch},
  title     = {Exact Exponential Algorithms},
  series    = {Texts in Theoretical Computer Science. An {EATCS} Series},
  publisher = {Springer},
  year      = {2010},
  doi       = {10.1007/978-3-642-16533-7}
}

@inproceedings{conf/ijcai/0001S025,
  author    = {Junqiang Peng and Zimo Sheng and Mingyu Xiao},
  title     = {New Algorithms for {\#}2-SAT and {\#}3-SAT},
  booktitle = {Proceedings of the 34th International Joint Conference on Artificial Intelligence ({IJCAI} 2025)},
  pages     = {2666--2674},
  year      = {2025},
  doi       = {10.24963/IJCAI.2025/297}
}

@inproceedings{conf/lics/BannachDGH25,
  author    = {Max Bannach and Erik D. Demaine and Timothy Gomez and Markus Hecher},
  title     = {{\#}P is Sandwiched by One and Two {\#}2DNF Calls: Is Subtraction Stronger Than We Thought?},
  booktitle = {Proceedings of the 40th Annual {ACM/IEEE} Symposium on Logic in Computer Science ({LICS} 2025)},
  pages     = {31--43},
  year      = {2025},
  doi       = {10.1109/LICS65433.2025.00010}
}

@inproceedings{conf/sosa/BannachDGH26,
  author    = {Max Bannach and Erik D. Demaine and Timothy Gomez and Markus Hecher},
  title     = {A Novel Reduction from {\#}SAT to {\#}2SAT Based on Symmetry: \emph{Simply Drop the Large Clauses}},
  booktitle = {Proceedings of the 2026 Symposium on Simplicity in Algorithms ({SOSA} 2026)},
  pages     = {254--265},
  year      = {2026},
  doi       = {10.1137/1.9781611978964.19}
}

@article{journals/jcss/CalabroIKP08,
  author    = {Chris Calabro and Russell Impagliazzo and Valentine Kabanets and Ramamohan Paturi},
  title     = {The complexity of Unique k-SAT: An Isolation Lemma for k-CNFs},
  journal   = {J. Comput. Syst. Sci.},
  volume    = {74},
  number    = {3},
  pages     = {386--393},
  year      = {2008},
  doi       = {10.1016/J.JCSS.2007.06.015}
}

@inproceedings{conf/iwpec/Traxler08,
  author    = {Patrick Traxler},
  title     = {The Time Complexity of Constraint Satisfaction},
  booktitle = {Proceedings of the 3rd International Workshop on Parameterized and Exact Computation ({IWPEC} 2008)},
  pages     = {190--201},
  year      = {2008},
  doi       = {10.1007/978-3-540-79723-4\_18}
}

@article{journals/talg/CyganDLMNOPSW16,
  author    = {Marek Cygan and Holger Dell and Daniel Lokshtanov and D{\'{a}}niel Marx and Jesper Nederlof and Yoshio Okamoto and Ramamohan Paturi and Saket Saurabh and Magnus Wahlstr{\"{o}}m},
  title     = {On Problems as Hard as {CNF-SAT}},
  journal   = {{ACM} Trans. Algorithms},
  volume    = {12},
  number    = {3},
  pages     = {41:1--41:24},
  year      = {2016},
  doi       = {10.1145/2925416}
}

@article{journals/jda/MonienP06,
  author    = {Burkhard Monien and Robert Preis},
  title     = {Upper bounds on the bisection width of 3- and 4-regular graphs},
  journal   = {J. Discrete Algorithms},
  volume    = {4},
  number    = {3},
  pages     = {475--498},
  year      = {2006},
  doi       = {10.1016/J.JDA.2005.12.009}
}

@article{journals/ipl/FominH06,
  author    = {Fedor V. Fomin and Kjartan H{\o}ie},
  title     = {Pathwidth of cubic graphs and exact algorithms},
  journal   = {Inf. Process. Lett.},
  volume    = {97},
  number    = {5},
  pages     = {191--196},
  year      = {2006},
  doi       = {10.1016/J.IPL.2005.10.012}
}

@article{journals/talg/GaspersS17,
  author    = {Serge Gaspers and Gregory B. Sorkin},
  title     = {Separate, Measure and Conquer: Faster Polynomial-Space Algorithms for Max 2-CSP and Counting Dominating Sets},
  journal   = {{ACM} Trans. Algorithms},
  volume    = {13},
  number    = {4},
  pages     = {44:1--44:36},
  year      = {2017},
  doi       = {10.1145/3111499}
}

@article{journals/jacm/FominGK09,
  author    = {Fedor V. Fomin and Fabrizio Grandoni and Dieter Kratsch},
  title     = {A measure {\&} conquer approach for the analysis of exact algorithms},
  journal   = {J. {ACM}},
  volume    = {56},
  number    = {5},
  pages     = {25:1--25:32},
  year      = {2009},
  doi       = {10.1145/1552285.1552286}
}

@article{journals/tcs/ChuXZ21,
  author    = {Huairui Chu and Mingyu Xiao and Zhe Zhang},
  title     = {An improved upper bound for {SAT}},
  journal   = {Theor. Comput. Sci.},
  volume    = {887},
  pages     = {51--62},
  year      = {2021},
  doi       = {10.1016/J.TCS.2021.06.045}
}

@article{journals/iandc/PengX23,
  author    = {Junqiang Peng and Mingyu Xiao},
  title     = {Further improvements for {SAT} in terms of formula length},
  journal   = {Inf. Comput.},
  volume    = {294},
  pages     = {105085},
  year      = {2023},
  doi       = {10.1016/J.IC.2023.105085}
}

@article{journals/jacm/PaturiPSZ05,
  author    = {Ramamohan Paturi and Pavel Pudl{\'{a}}k and Michael E. Saks and Francis Zane},
  title     = {An improved exponential-time algorithm for \emph{k}-SAT},
  journal   = {J. {ACM}},
  volume    = {52},
  number    = {3},
  pages     = {337--364},
  year      = {2005},
  doi       = {10.1145/1066100.1066101}
}

@inproceedings{conf/focs/PaturiPZ97,
  author    = {Ramamohan Paturi and Pavel Pudl{\'{a}}k and Francis Zane},
  title     = {Satisfiability Coding Lemma},
  booktitle = {Proceedings of the 38th Annual Symposium on Foundations of Computer Science ({FOCS} 1997)},
  pages     = {566--574},
  year      = {1997},
  doi       = {10.1109/SFCS.1997.646146}
}

@inproceedings{conf/focs/Schoning99,
  author    = {Uwe Sch{\"{o}}ning},
  title     = {A Probabilistic Algorithm for k-SAT and Constraint Satisfaction Problems},
  booktitle = {Proceedings of the 40th Annual Symposium on Foundations of Computer Science ({FOCS} 1999)},
  pages     = {410--414},
  year      = {1999},
  doi       = {10.1109/SFFCS.1999.814612}
}

@inproceedings{conf/coco/CalabroIP06,
  author    = {Chris Calabro and Russell Impagliazzo and Ramamohan Paturi},
  title     = {A Duality between Clause Width and Clause Density for {SAT}},
  booktitle = {Proceedings of the 21st Annual {IEEE} Conference on Computational Complexity ({CCC} 2006)},
  pages     = {252--260},
  year      = {2006},
  doi       = {10.1109/CCC.2006.6}
}

@incollection{series/faia/DantsinH21,
  author    = {Evgeny Dantsin and Edward A. Hirsch},
  title     = {Worst-Case Upper Bounds},
  booktitle = {Handbook of Satisfiability - Second Edition},
  publisher    = {{IOS} Press},
  pages     = {669--692},
  year      = {2021},
  doi       = {10.3233/FAIA200999}
}

@article{journals/siamcomp/Iwama89,
  author    = {Kazuo Iwama},
  title     = {{CNF} Satisfiability Test by Counting and Polynomial Average Time},
  journal   = {{SIAM} J. Comput.},
  volume    = {18},
  number    = {2},
  pages     = {385--391},
  year      = {1989},
  doi       = {10.1137/0218026}
}

@article{journals/ipl/Lozinskii92,
  author    = {Eliezer L. Lozinskii},
  title     = {Counting Propositional Models},
  journal   = {Inf. Process. Lett.},
  volume    = {41},
  number    = {6},
  pages     = {327--332},
  year      = {1992},
  doi       = {10.1016/0020-0190(92)90160-W}
}

@inproceedings{conf/focs/Scheder21,
  author    = {Dominik Scheder},
  title     = {{PPSZ} is better than you think},
  booktitle = {Proceedings of the 62nd {IEEE} Annual Symposium on Foundations of Computer Science ({FOCS} 2021)},
  pages     = {205--216},
  year      = {2021},
  doi       = {10.1109/FOCS52979.2021.00028}
}

@inproceedings{conf/icalp/Liu18,
  author    = {Sixue Liu},
  title     = {Chain, Generalization of Covering Code, and Deterministic Algorithm for \emph{k}-{SAT}},
  booktitle = {Proceedings of the 45th International Colloquium on Automata, Languages, and Programming ({ICALP} 2018)},
  pages     = {88:1--88:13},
  year      = {2018},
  doi       = {10.4230/LIPIcs.ICALP.2018.88}
}

@article{journals/iandc/ArvindK06,
  author    = {Vikraman Arvind and Piyush P. Kurur},
  title     = {Graph Isomorphism is in {SPP}},
  journal   = {Inf. Comput.},
  volume    = {204},
  number    = {5},
  pages     = {835--852},
  year      = {2006},
  doi       = {10.1016/J.IC.2006.02.002}
}

@inproceedings{conf/focs/Valiant06,
  author    = {Leslie G. Valiant},
  title     = {Accidental Algorithms},
  booktitle = {Proceedings of the 47th Annual {IEEE} Symposium on Foundations of Computer Science ({FOCS} 2006)},
  pages     = {509--517},
  year      = {2006},
  doi       = {10.1109/FOCS.2006.7}
}

@inproceedings{conf/icalp/AbboudFW20,
  author    = {Amir Abboud and Shon Feller and Oren Weimann},
  title     = {On the Fine-Grained Complexity of Parity Problems},
  booktitle = {Proceedings of the 47th International Colloquium on Automata, Languages, and Programming ({ICALP} 2020)},
  pages     = {5:1--5:19},
  year      = {2020},
  doi       = {10.4230/LIPICS.ICALP.2020.5}
}

@article{journals/algorithmica/GoldbergR24,
  author    = {Leslie Ann Goldberg and Marc Roth},
  title     = {Parameterised and Fine-Grained Subgraph Counting, Modulo 2},
  journal   = {Algorithmica},
  volume    = {86},
  number    = {4},
  pages     = {944--1005},
  year      = {2024},
  doi       = {10.1007/S00453-023-01178-0}
}

@inproceedings{conf/esa/CurticapeanDH21,
  author    = {Radu Curticapean and Holger Dell and Thore Husfeldt},
  title     = {Modular Counting of Subgraphs: Matchings, Matching-Splittable Graphs, and Paths},
  booktitle = {Proceedings of the 29th Annual European Symposium on Algorithms ({ESA} 2021)},
  pages     = {34:1--34:17},
  year      = {2021},
  doi       = {10.4230/LIPICS.ESA.2021.34}
}

@inproceedings{DBLP:conf/walcom/Hoi0S024,
  author    = {Gordon Hoi and Sanjay Jain and Ammar Fathin Sabili and Frank Stephan},
  title     = {A Bisection Approach to Subcubic Maximum Induced Matching},
  booktitle = {Proceedings of the 18th International Conference and Workshops on Algorithms and Computation ({WALCOM} 2024)},
  pages     = {257--272},
  year      = {2024},
  doi       = {10.1007/978-981-97-0566-5\_19}
}

@inproceedings{DBLP:conf/cp/HoiJS20,
  author    = {Gordon Hoi and Sanjay Jain and Frank Stephan},
  title     = {A Faster Exact Algorithm to Count {X3SAT} Solutions},
  booktitle = {Proceedings of the 26th International Conference on Principles and Practice of Constraint Programming ({CP} 2020)},
  pages     = {375--391},
  year      = {2020},
  doi       = {10.1007/978-3-030-58475-7\_22}
}

@inproceedings{conf/focs/BjorklundH13,
  author    = {Andreas Bj{\"{o}}rklund and Thore Husfeldt},
  title     = {The Parity of Directed Hamiltonian Cycles},
  booktitle = {Proceedings of the 54th Annual {IEEE} Symposium on Foundations of Computer Science ({FOCS} 2013)},
  pages     = {727--735},
  year      = {2013},
  doi       = {10.1109/FOCS.2013.83}
}

@inproceedings{conf/icalp/BjorklundDH15,
  author    = {Andreas Bj{\"{o}}rklund and Holger Dell and Thore Husfeldt},
  title     = {The Parity of Set Systems Under Random Restrictions with Applications to Exponential Time Problems},
  booktitle = {Proceedings of the 42nd International Colloquium on Automata, Languages, and Programming ({ICALP} 2015)},
  pages     = {231--242},
  year      = {2015},
  doi       = {10.1007/978-3-662-47672-7\_19}
}

@article{journals/talg/DellHMTW14,
  author       = {Holger Dell and
                  Thore Husfeldt and
                  D{\'{a}}niel Marx and
                  Nina Taslaman and
                  Martin Wahlen},
  title        = {Exponential Time Complexity of the Permanent and the Tutte Polynomial},
  journal      = {{ACM} Trans. Algorithms},
  volume       = {10},
  number       = {4},
  pages        = {21:1--21:32},
  year         = {2014},
  doi          = {10.1145/2635812}
}

@article{journals/corr/abs-2201-03220,
  author       = {Gordon Hoi and
                  Ammar Fathin Sabili and
                  Frank Stephan},
  title        = {An Exact Algorithm for finding Maximum Induced Matching in Subcubic
                  Graphs},
  journal      = {CoRR},
  volume       = {abs/2201.03220},
  year         = {2022},
  url          = {https://arxiv.org/abs/2201.03220},
  eprinttype   = {arXiv},
  eprint       = {2201.03220}
}
